\documentclass[11pt,draftcls,onecolumn]{IEEEtran} 
\pdfoutput=1
\usepackage{setspace}

\usepackage{cite}

\usepackage{amssymb}
\usepackage{amsmath}

\usepackage{amsthm}
\usepackage{graphicx}
\graphicspath{{./graphics/}{.}}
\usepackage{color}

\usepackage{multirow}

\newcommand{\Complex}{\mathbb{C}}
\newcommand{\hd}{\hdots}

\long\def\symbolfootnote[#1]#2{\begingroup
\def\thefootnote{\fnsymbol{footnote}}\footnote[#1]{#2}\endgroup}
\newcommand{\Expected}{\mathbb{E}}

\newcommand{\bsm}[1]{{\boldsymbol #1}}

\newtheorem*{thm}{Theorem 2}

\newtheorem*{thmx}{Theoem 1 {(Newey and McFadden\cite[Ch. 36, Theorem 2.1]{handbook_econometrics}})}

\newtheorem{lem}{Lemma}

\IEEEoverridecommandlockouts
\begin{document}

\title {Blind Channel Estimation for  Amplify-and-Forward Two-Way Relay Networks Employing $M$-PSK Modulation}
\author{Saeed Abdallah and Ioannis N. Psaromiligkos\thanks{The authors are with McGill University, Department of Electrical and Computer Engineering, 3480 University Street, Montreal, Quebec, H3A 2A7, Canada. Email: saeed.abdallah@mail.mcgill.ca; yannis@ece.mcgill.ca, phone: +1 (514) 398-2465, fax: +1 (514) 398-4470. This work was supported in part by the Natural Sciences and Engineering Research Council of Canada (NSERC) under grant 262017.}\thanks{This work was presented in part at the 2011 IEEE International Conference on Acoustics, Speech and Signal Processing (ICASSP), Prague, Czech Republic}}

\maketitle
\begin{abstract}
We consider the problem of channel estimation for amplify-and-forward (AF) two-way relay networks (TWRNs). Most works on this problem focus on pilot-based approaches which impose a significant training overhead that reduces the spectral efficiency of the system. To avoid such losses, this work proposes blind channel estimation algorithms for AF TWRNs that employ constant-modulus (CM) signaling.  Our main algorithm is based on the deterministic maximum likelihood (DML) approach. Assuming \textit{M}-PSK modulation, we show that the resulting estimator is consistent and approaches the true channel with high probability at high SNR for modulation orders higher than 2. For BPSK, however, the DML performs poorly and we propose an alternative algorithm that performs much better by taking into account the BPSK structure of the data symbols. For comparative purposes, we also investigate the Gaussian maximum-likelihood (GML) approach which treats the data symbols as Gaussian-distributed nuisance parameters. We derive the Cramer-Rao bound and use Monte-Carlo simulations to investigate the mean squared error (MSE) performance of the proposed algorithms. We also compare the symbol-error rate (SER) performance of the DML algorithm with that of the training-based least-squares (LS) algorithm and demonstrate that the DML offers a superior tradeoff between accuracy and spectral efficiency.

\begin{center}\textbf{\textit{Index Terms}:} Amplify and Forward, Channel Estimation, Constant Modulus, Cramer-Rao bound, Maximum-likelihood, Two-way Relays.\end{center}
\end{abstract}

\section{Introduction}

Amplify-and-forward (AF) two-way relay networks (TWRNs)~\cite{rankov07} have recently emerged as a spectrally efficient approach for bidirectional communication between two terminals. The majority of works on TWRNs assume the availability of channel knowledge at the terminals, which makes it important to develop efficient channel estimation algorithms. Although some recent works have considered noncoherent two-way relaying which avoids channel estimation by employing differential modulation schemes (c.f.~\cite{song2010differential,cui2009differential}), the resulting systems have lower performance than those employing coherent schemes. 

So far, researchers have taken a training-based approach to channel estimation for AF TWRNs, estimating the channels using pilot symbols known to both terminals. In~\cite{gao2009optimal}, the training-based maximum-likelihood (ML) estimator was developed for the classical single-antenna single-relay system assuming flat-fading channels. In~\cite{feifei_power}, the ML approach was used to acquire initial channel estimates at the relay that were employed to denoise the training signal. Power was then allocated at the relay to optimize channel estimation at the terminals. The case when the relay and the two terminals are equipped with multiple antennas was considered in~\cite{tensor_based}, where initial estimates of the multi-input multi-output (MIMO) channels were acquired at each terminal using the training-based least-squares (LS) algorithm, and then a tensor-based approach was employed to improve the estimation accuracy by exploiting the structure of the channel matrices. Channel estimation for TWRNs in frequency selective environments has also been considered in a number of works. LS estimation for OFDM systems was considered in~\cite{gao2009channel}. A nulling-based LS algorithm was proposed in~\cite{joint_CFO} to jointly estimate the channel and the carrier-frequency offset (CFO). Finally, channel estimation and the design of training symbols for TWRNs operating in time-varying environments has been considered in~\cite{time_varying}, where a complex-exponential basis expansion model (CE-BEM) was used to facilitate LS channel estimation.

The obvious shortcoming of the pilot-based approach is the training overhead which in turn reduces the spectral efficiency of TWRNs. Blind channel estimation would avoid this burden, which makes it a more spectrally-efficient approach. Little effort, however, has so far been dedicated to develop blind channel estimation algorithms for TWRNs. In~\cite{feifei-blind}, a blind channel estimator was proposed for OFDM-based TWRNs by using the second-order statistics of the received signal and applying linear precoding at both terminals, while in ~\cite{joint_detection} a semi-blind approach to joint channel estimation and detection for MIMO-OFDM TWRNs was developed using the expectation conditional maximization (ECM) algorithm.

In this work, we consider blind channel estimation for TWRNs that employ constant-modulus (CM) signaling. Because of its constant envelope, CM signaling permits the use of inexpensive and energy-efficient nonlinear amplifiers. In fact, CM signaling in the form of continuous-phase modulation (CPM) is employed in the well-known GSM cellular standard, while 8-PSK modulation is employed in the Enhanced Data Rates for GSM Evolution (EDGE)~\cite{edge1999}. Moreover, QPSK modulation is supported in the 3rd Generation Partnership Program (3GPP) Long Term Evolution (LTE) and LTE-Advanced wireless standards~\cite{LTE_CM}. CM signaling has also been employed in a number of works on amplify-and-forward relay networks~\cite{maw2008,chowdhery2009,zhang2010,dharmawansa2010}. In fact, a novel relaying protocol has been recently proposed in~\cite{yang2011} whereby CM signalling is employed both at the terminals and at the relay which only forwards the phase of the received signal to the destination. This protocol is referred to as phase-only forward relaying.

Assuming nonreciprocal flat-fading channels, we propose a deterministic ML (DML)-based algorithm that estimates the channel blindly by treating the data symbols as deterministic unknowns. While the proposed algorithm may be applied to any type of CM signaling, we analyze its asymptotic performance assuming that the terminals employ $M$-PSK modulation. Noting that consistency is not guaranteed for ML estimators when the data symbols are treated as deterministic unknowns~\cite{neyman1948,kay96}, we prove that our DML estimator is consistent when the channel parameters belong to compact sets. We also study the asymptotic behavior of the DML estimator at high SNR and prove that it approaches the true channel with high probability for modulation orders higher than 2. For $M=2$, however, the DML estimator performs poorly, and we propose an alternative estimator based on the constrained ML (CML) approach which provides much better performance by explicitly taking into account the BPSK structure of the data symbols. As a simple alternative to the DML approach, we also consider the Gaussian ML (GML) estimator which can be obtained by treating the data symbols as Gaussian-distributed nuisance parameters. When CM signaling is employed, the GML estimator takes the form of a sample average which is consistent and can be updated online but suffers from an error floor at high SNR. To evaluate the performance of our estimator, we also derive two Cramer-Rao bounds (CRBs) for our estimation problem. The first bound is obtained by treating the data symbols as deterministic unknown parameters and the second is the modified CRB~\cite{gini2000} which treats the data symbols as random nuisance parameters.

Monte Carlo simulations are used to investigate the performance of the proposed algorithms. For $M>2$, we show that the DML estimator outperforms the GML estimator at medium-to-high SNR and approaches the CRB at high SNR. For $M=2$ we show that the CML-inspired estimator outperforms the GML estimator except at very low SNR. We also investigate the tradeoffs of following the blind approach by comparing the symbol-error rate (SER) performance of the DML estimator with that of the training-based LS estimator. We show that the DML approach provides a better tradeoff between accuracy and spectral efficiency.

The remainder of the paper is organized as follows. In Section~\ref{system_model}, we present our system model. In Section~\ref{proposed_algorithms} we present the proposed algorithms. In Section~\ref{performance_analysis}, we analyze the asymptotic behavior of our estimators and derive the CRB. We show our simulation results and comparisons in Section~\ref{simulations}. Finally, our conclusions are presented in Section~\ref{conclusions}.

\vspace{-2ex}
\begin{figure}[htbp]
\centering

\includegraphics[width=3.5in]{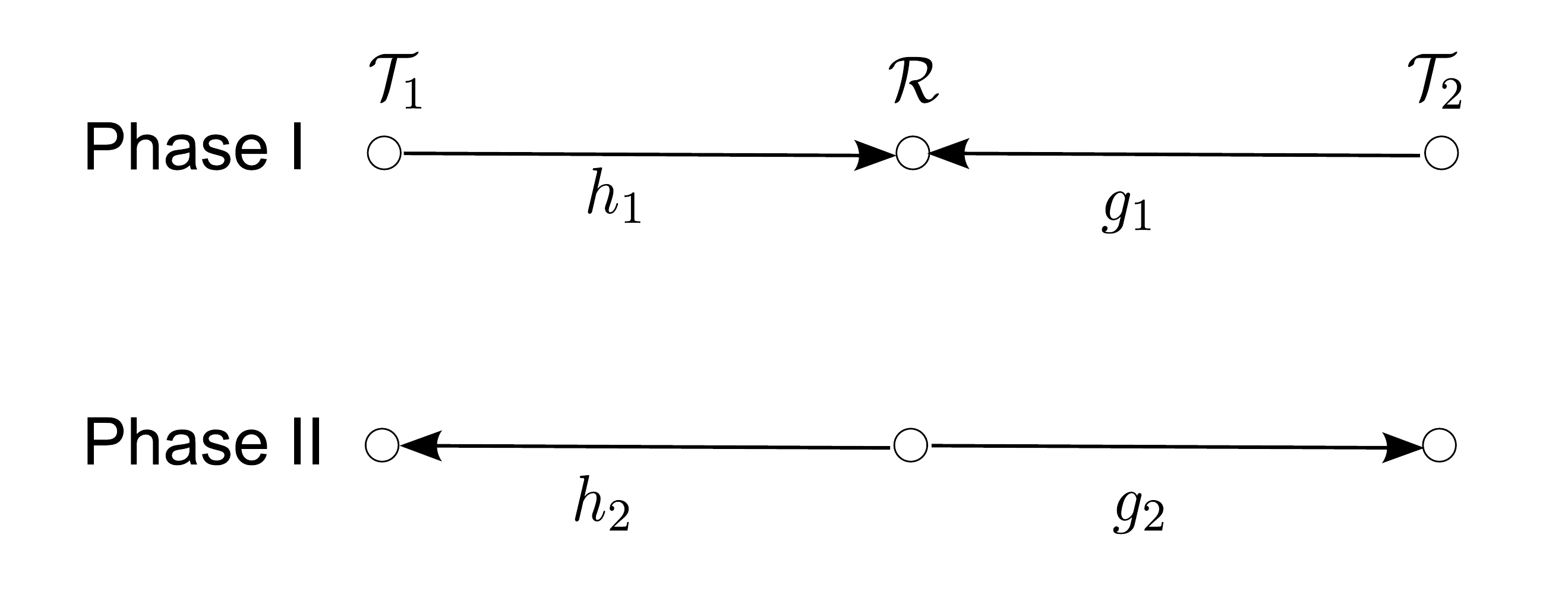}
\caption{The two-way relay network with two source terminals and one relay node.}
\label{two_way_nonreciprocal}
\end{figure}
\vspace{-5ex}
\section{System Model}
\label{system_model}
We consider the typical half-duplex TWRN with two source nodes, $\mathcal{T}_1$ and $\mathcal{T}_2$, and a single relaying node $\mathcal{R}$, shown in Fig.~\ref{two_way_nonreciprocal}. The network operates in quasi-static flat-fading channel conditions. Each data transmission period is divided into two phases. In the first phase, $\mathcal{T}_1$ and $\mathcal{T}_2$ simultaneously transmit to $\mathcal{R}$ the $M$-PSK data symbols $t_{1}$ and $t_{2}$, respectively. The symbols $t_1$ and $t_2$ are of the form $t_{1}=\sqrt{P_1}e^{\jmath\phi_{1}}$ and $t_{2}=\sqrt{P_2}e^{\jmath\phi_{2}}$, where $P_1$ and $P_2$ are the transmission powers of $\mathcal{T}_1$, $\mathcal{T}_2$, respectively, $\phi_{1}$ and $\phi_{2}$ are the information-bearing phases randomly and independently chosen from the set $S_M=\{(2\ell-1)\pi/M,\ell=1,\hdots,M\}$, and $\jmath\triangleq\sqrt{-1}$. The received signal at the relay during the first transmission phase is given by
\begin{equation} 
\label{received_signal}
r=h_1t_{1}+g_1t_{2}+n,
\end{equation}
where $h_1$ and $g_1$ are the complex coefficients of the flat-fading channels $\mathcal{T}_1\rightarrow\mathcal{R}$ and $\mathcal{T}_2\rightarrow\mathcal{R}$, respectively, and $n$ is zero-mean complex additive white Gaussian noise (AWGN) with variance $\sigma_r^2$. In the second phase, $\mathcal{R}$ broadcasts the amplified signal $Ar$, where $A$ is the amplification factor, assumed to be known at both terminals as is common practice (cf.~\cite{gao2009optimal,feifei_power}). The amplified signal passes through the channels $h_2$ and $g_2$ to reach $\mathcal{T}_1$ and $\mathcal{T}_2$, respectively. The complex channel coefficients $h_1$, $h_2$, $g_1$ and $g_2$ remain fixed during the estimation period. To maintain an average power of $P_r$ at the relay over the long term, the amplification factor is chosen as~\cite{gao2009optimal} 
\begin{equation}
A=\sqrt{\frac{P_r}{P_1+P_2+\sigma_r^2}}.
\end{equation}
Without loss of generality, we consider channel estimation at node $\mathcal{T}_1$. The received signal at $\mathcal{T}_1$ in the second transmission phase is
\begin{equation}
\label{atT1}
z=Ah_1h_2t_{1}+Ag_1h_2t_{2}+Ah_2n+\eta,
\end{equation}
where $\eta$ is the zero-mean AWGN term at $\mathcal{T}_1$ with variance $\sigma_t^2$. We assume that the noise variances $\sigma_r^2$ and $\sigma_t^2$ are unknown at $\mathcal{T}_1$.
The term $Ah_1h_2t_1$ represents the self-interference at $\mathcal{T}_1$. In~\eqref{atT1}, there are 6 unknown real channel parameters: the magnitudes and phases of the complex parameters $h_1$, $h_2$, and $g_1$. However, not all of them are identifiable. By inspecting~\eqref{atT1}, we see that the channel parameters which can be identified (blindly or otherwise) when the noise variances are unknown are\footnote{In this paper, $|\cdot|$ and $\angle(\cdot)$ denote the magnitude and phase of a complex number, respectively.} $|h_1||h_2|$, $|g_1||h_2|$, $\angle{h_1}+\angle{h_2}$, and $\angle{g_1}+\angle{h_2}$. On the other hand, it is not possible to identify the individual magnitudes $|h_1|, |h_2|$ and $|g_1|$, or individual phases $\angle{h_1}$, $\angle{h_2}$ and $\angle{g_1}$. In other words, the identifiable channel parameters are $a\triangleq h_1h_2$ and $b\triangleq g_1h_2$. 
Under the CM assumption, it is sufficient for detection purposes to know $a$ and $\phi_b\triangleq\angle{b}$.

\section{Proposed Channel Estimation Algorithms}
\label{proposed_algorithms}

Estimation is performed at $\mathcal{T}_1$ using $N$ received samples, $z_i$, $i=1,\hdots,N,$ of the form given by~\eqref{atT1}. The time index $i$, $i=1,\hdots,N,$ is used to indicate the realizations of $t_1$, $t_2$, $\phi_1$, $\phi_2$, $n$, $\eta$, that gave rise to each sample $z_i$. Let\footnote{In this paper, $(\cdot)^{T}$ and $(\cdot)^{*}$ are used to denote the transpose and complex conjugate operations, respectively.} $\boldsymbol{z}\triangleq[z_1,\hdots,z_N]^{T}$ be the vector of received samples. This vector can be expanded as
\begin{equation}
\label{received_vector}
\boldsymbol{z}=Aa\boldsymbol{t}_1+Ab\boldsymbol{t}_2+Ah_2{\boldsymbol n}+\boldsymbol{\eta},
\end{equation}
where ${\boldsymbol t}_1\triangleq[t_{11},\hdots,t_{1N}]^{T}$, ${\boldsymbol t}_2\triangleq[t_{21},\hdots,t_{2N}]^{T}$, ${\boldsymbol n}\triangleq[n_1,\hdots,n_N]^{T}$ and $\boldsymbol{\eta}\triangleq[\eta_1,\hdots,\eta_N]^{T}$.

\subsection{Deterministic ML Approach}
To avoid dealing with a complicated likelihood function, we treat the transmitted symbols $t_{2i}, i=1,\hdots,N$ as deterministic unknowns. We also ignore the finite alphabet constraint that restricts the phases $\phi_{2i}$, $i=1\hd,N$ to the set $S_{M}$. Nonetheless, the actual statistics of these phases will be used in Section~\ref{performance_analysis} to analyze the behavior of the resulting estimator. Due to the above assumptions, the DML approach can be used to blindly estimate the sums $\phi_{bi}\triangleq\phi_{2i}+\phi_b, i=1,\hdots,N,$ but it cannot be used to obtain separate estimates of $\phi_b$ and $\phi_{2i}$, $i=1\hd, N$. However, since $M$-PSK modulation is assumed, the Viterbi-Viterbi algorithm~\cite{viterbi83} may be used to blindly estimate $\phi_b$. A few pilots, however, are still needed to resolve the resulting $M$-fold ambiguity in $\phi_b$.

We define the vector of unknown parameters for the DML algorithm as $\boldsymbol{\theta}\triangleq[a, |b|, \phi_{b1},\hdots,\phi_{bN}, \sigma_o^2]^{T}$, where $\sigma_o^2\triangleq A^2|h_2|^2\sigma_r^2+\sigma_t^2$ is the overall noise variance at $\mathcal{T}_1$. With $\boldsymbol{t}_2$ assumed deterministic and $\boldsymbol{t}_1$ known, the received vector $\boldsymbol{z}$ is complex Gaussian with $\Expected\{\boldsymbol{z}\}=Aa\boldsymbol{t}_1+Ab\boldsymbol{t}_2$ and covariance $\Expected\{\boldsymbol{z}\boldsymbol{z}^{H}\}=\sigma_o^2\boldsymbol{I}$. Hence, the log-likelihood function is given by
\begin{equation} 
\label{log_likelihood}
\begin{split}
\mathcal{L}({\boldsymbol z};\boldsymbol{\theta})&=-\frac{1}{\sigma_o^2}\Vert {\boldsymbol z}-Aa{\boldsymbol t}_1-Ab{\boldsymbol t}_2\Vert^2-N\log\left(\pi\sigma_o^2\right)\\
&=-\frac{1}{\sigma_o^2}\sum\limits_{i=1}^{N}\left|z_{i}-Aat_{1i}-A|b|\sqrt{P_2}e^{\jmath\phi_{bi}}\right|^2-N\log\left(\pi\sigma_o^2\right).
\end{split}
\end{equation}
Let $\hat{a}$, $\widehat{|b|}$, $\widehat{\sigma_o^2}$, $\hat{\phi}_{bi}$ be the ML estimates of $a$, $|b|$, $\sigma_o^2$ and $\phi_{bi}$, $i=1,\hdots,N,$ respectively. It is straightforward to show that $\hat{\phi}_{bi}=\angle(z_{i}-A\hat{a}t_{1i})$, which when substituted into $\mathcal{L}({\boldsymbol z};\boldsymbol{\theta})$ yields the updated log-likelihood
\begin{equation}
\label{updated_likelihood}
\begin{split}
\mathcal{L}'({\boldsymbol z};a,|b|,\sigma_o^2)&=-\frac{1}{\sigma_o^2}\sum\limits_{i=1}^{N}\bigg(|z_{i}-Aat_{1i}|-A|b|\sqrt{P_2} \bigg)^2-N\log\left(\pi\sigma_o^2\right).
\end{split}
\end{equation}
Maximizing~\eqref{updated_likelihood} with respect to $|b|$, we get 
\begin{equation}
\label{b_estimator}
\widehat{|b|}=\frac{\sum\limits_{i=1}^{N}\left|z_i-A\hat{a}t_{1i}\right|}{NA\sqrt{P_2}}.
\end{equation}
Substituting $\widehat{|b|}$ in place of $|b|$ in~\eqref{updated_likelihood}, we obtain
\begin{equation}
\label{a_h2_ML}
\mathcal{L}''({\boldsymbol z};a,\sigma_o^2)=-\frac{1}{\sigma_o^2}\sum\limits_{i=1}^{N}\left(|z_i-Aat_{1i}|-\frac{1}{N}\sum\limits_{k=1}^{N}|z_k-Aat_{1k}|\right)^2-N\log\left(\pi\sigma_o^2\right).
\end{equation}
Hence, the ML estimate of $a$ is given by\footnote{We note that a slightly different estimator may be obtained by using a likelihood function that also takes into account the pilot symbols used to resolve the ambiguity in $\phi_b$. 
Since the required number of pilots is very small, the resulting estimator will only be slightly better than~\eqref{a_ML}.}
\begin{equation}
\label{a_ML}
\begin{split}
\hat{a}&=\arg\ \min_{u\in\Complex}\Bigg\{\sum\limits_{i=1}^{N}\left(|z_i-Aut_{1i}|-\frac{1}{N}\sum\limits_{k=1}^{N}|z_k-Aut_{1k}|\right)^2\Bigg\}.
\end{split}
\end{equation}
Finally, the ML estimate of the parameter $\sigma_o^2$ is neither needed for detection nor for estimating $a$, $|b|$, but is provided for completeness and is given by
\begin{equation}
\widehat{\sigma_o^2}=\frac{1}{N}\sum\limits_{i=1}^{N}\left(|z_i-A\hat{a}t_{1i}|-\frac{1}{N}\sum\limits_{k=1}^{N}|z_k-A\hat{a}t_{1k}|\right)^2.
\end{equation}
We note that the parameter $|b|$ is also not needed for detection when CM signaling is used. 

The ML estimator in~\eqref{a_ML} has an intuitive interpretation. Let $\tilde{z}_i(u)\triangleq z_i-Aut_{1i}, i=1,\hdots,N$ be the ``cleaned''  versions of the received signal samples after self-interference has been removed using the complex value $u$ as an estimate of $a$. The signals $\tilde{z}_i(u), i=1,\hdots,N$ are independently-generated realizations of the RV $\tilde{z}(u)\triangleq A(a-u)t_1+Abt_2+Ah_2n+\eta$. The quantity
\begin{equation}
\label{random_sequence0}
W_N(u)=\frac{1}{N-1}\sum\limits_{i=1}^{N}\left(|\tilde{z}_i(u)|-\frac{1}{N}\sum\limits_{k=1}^{N}|\tilde{z}_k(u)|\right)^2
\end{equation}
which is the objective function in~\eqref{a_ML}, scaled by $\frac{1}{N-1}$, also represents the sample-variance of the envelope of $\tilde{z}(u)$. We demonstrate in Section~\ref{performance_analysis} that the value $u=a$ which completely cancels the interference also minimizes the variance of the envelope of $\tilde{z}(u)$. Hence, the variance of the envelope of $\tilde{z}(u)$ may be seen as a measure of the level of self-interference. 

By treating the transmitted symbols $t_{2i}, i=1,\hd,N$ as deterministic unknowns, the estimator in~\eqref{a_ML} ignores the underlying structure of the phases $\phi_{2i}, i=1,\hd,N$. A valid question is whether a better blind estimator may be obtained by explicitly taking into account the structure of the phases. This question is particularly relevant for BPSK modulation, for which we show in Section~\ref{performance_analysis} that the DML estimator in~\eqref{a_ML} performs poorly as it experiences infinitely many global minima at high SNR. There are two alternative blind ML approaches which take into account the phase structure of the data. The first approach is to perform ML estimation using the likelihood function $f(z\ |\  a, b, \sigma_o^2)$ which can be obtained by treating $t_{2i}, i=1,\hd, N$ as random parameters and marginalizing over them. Unfortunately, this approach is very complicated due to the mixed-Gaussian nature of $f(z\ |\  a, b, \sigma_o^2)$. The other alternative is to follow a constrained ML (CML) approach by solving the likelihood function in~\eqref{log_likelihood} subject to the constraint that $\phi_{2i}\in S_M, i=1,\hd,N$. For higher order modulations, it is difficult to apply this approach and still obtain a closed-form objective function. However, as we shall see next, this approach becomes more feasible for BPSK modulation.

\subsection{Constrained ML Approach for BPSK}

As we will see shortly, a straightforward application of the CML approach for $M=2$ is not feasible as it results in a 3-dimensional search. However, it is possible to obtain a CML-inspired estimator which possesses a closed-form objective function by utilizing just a few pilot symbols to estimate the phase $\phi_b$. 

Under BPSK modulation, the CML approach minimizes the same log-likelihood function as in~\eqref{log_likelihood}, but subject to the constraint that $t_{2i}=\pm\sqrt{P_2},\ i=1,\hd, N$. We focus on the first term in~\eqref{log_likelihood} since the estimation of $\sigma_o^2$ is decoupled from the estimation of the remaining parameters. Let
\begin{equation}
\label{CML1}
\mathcal{J}(\bsm{z}; a, |b|, \phi_b, \bsm{t}_2)\triangleq \sum\limits_{i=1}^{N}\left|z_{i}-Aat_{1i}-A|b|\sqrt{P_2}e^{\jmath(\phi_{bi})}\right|^2
\end{equation}
be our objective function. The resulting estimate of $t_{2i}, i=1,\hd, N,$ is
\begin{equation}
\label{CML2}
\hat{t}_{2i}=\sqrt{P_2}\mbox{sgn}\left(\Re\{e^{-\jmath\phi_b}(z_i-Aat_{1i})\}\right).
\end{equation}
Substituting~\eqref{CML2} back into~\eqref{CML1}, we obtain
\begin{equation}
\label{CML3}
\mathcal{J}'(\bsm{z}; a, |b|, \phi_b)=\sum\limits_{i=1}^{N}|z_i-Aat_{1i}|^2+NA^2|b|^2P_2-2A|b|\sqrt{P_2}\sum\limits_{i=1}^{N}|\Re\{e^{-\jmath\phi_b}(z_i-Aat_{1i})\}|.
\end{equation}
From~\eqref{CML3}, we see that the CML estimate of $|b|$ is given by
\begin{equation}
\label{CML4}
\widehat{|b|}_c=\frac{1}{NA\sqrt{P_2}}\sum\limits_{i=1}^{N}|\Re\{e^{-\jmath\phi_b}(z_i-Aat_{1i})\}|.
\end{equation}
Substituting~\eqref{CML4} into~\eqref{CML3}, we obtain the updated objective function
\begin{equation}
\label{CML5}
\mathcal{J}''(\bsm{z}; a, \phi_b)=\sum\limits_{i=1}^{N}|z_i-Aat_{1i}|^2-\frac{1}{N}\left(\sum\limits_{i=1}^{N}|\Re\{e^{-\jmath\phi_b}(z_i-Aat_{1i})\}|\right)^2
\end{equation}
which we have to solve for $a$ and $\phi_b$. The CML estimate of $\phi_b$ in terms of $a$ is the solution of the following maximization problem
\begin{equation}
\label{CML6}
\hat{\phi}_{bc}=\arg\ \max_{\mathcal{\psi}\in[0, 2\pi)}\sum\limits_{i=1}^{N}|z_i-Aat_{1i}||\cos(\angle(z_i-Aat_{1i})-\psi)|.
\end{equation}
Unfortunately, the maximization in~\eqref{CML6} does not have a closed-form solution, which means that we cannot proceed to obtain an objective function that only depends on $a$, like the one in~\eqref{a_ML}. Thus, a strict application of the CML approach requires the use of a 3-dimensional search to estimate $a$ and $\phi_b$. 

To get around this problem, we propose to replace $\phi_b$ in~\eqref{CML5} with a pilot-based estimate of $\phi_b$ and then proceed to estimate $a$. We let $x_{1j}$, $x_{2j}$ and $s_j$, $j=1,\hd, J,$ be the $J$ pilot symbols transmitted by $\mathcal{T}_1$ and $\mathcal{T}_2$ and the corresponding samples received at $\mathcal{T}_1$, respectively. We may thus estimate $\phi_b$ by $\hat{\phi}_b=\angle\sum\limits_{j=1}^{J}(s_j-Aax_{1j})x_{2j}$. Substituting $\hat{\phi}_b$ into~\eqref{CML5}, we obtain the following estimate of $a$
\begin{equation}
\label{CML8}
\hat{a}_c=\arg\ \min_{u\in\Complex} \sum\limits_{i=1}^{N}|z_i-Aut_{1i}|^2-\frac{1}{N}\left(\sum\limits_{k=1}^{N}\left|\Re\left\{(z_k-Aut_{1k})e^{-\jmath\angle\sum\limits_{j=1}^{J}(s_j-Aux_{1j})x_{2j}}\right\}\right|\right)^2.
\end{equation}
We refer to the estimator in~\eqref{CML8} as the modified CML (MCML) estimator. 
\vspace{-2ex}
\subsection{Gaussian ML Approach}

In deriving the DML and the MCML estimators, we treated the data symbols $t_{2i}, i=1,\hd,N$ as deterministic unknowns. Another approach commonly used to deal with nuisance parameters is Gaussian ML estimation~\cite{Carvalho97}. In this case, the data symbols $t_{2i}, i=1,\hd, N$ are treated as i.i.d. complex Gaussian random variables with mean zero and variance $P_2$. Under this assumption, the total noise variance becomes $\sigma_g^2\triangleq P_2|b|^2+\sigma_o^2$. The resulting log-likelihood function on $a$ is
\begin{equation} 
\mathcal{L}(\bsm{z};a)=-\frac{1}{\sigma_g^2}\Vert {\boldsymbol z}-Aa{\boldsymbol t}_1\Vert^2-N\log(\pi\sigma_g^2).
\end{equation}
Hence, the GML estimate of a is
\begin{equation}
\label{GMLa0}
\hat{a}_g=\frac{\sum\limits_{i=1}^{N}t_{1i}^{*}z_i}{\sum\limits_{i=1}^{N}|t_{1i}|^2}=\frac{1}{NAP_1}\sum\limits_{i=1}^{N}t_{1i}^{*}z_i.
\end{equation}
Thus, the GML estimate of $a$ has the form of a computationally inexpensive sample-average. It can be easily updated at run-time by updating the average term as new samples arrive. The GML estimator will serve as a low-complexity benchmark with which to compare the performance of the proposed DML and MCML estimators.  

\section{Asymptotic Behavior Analysis}
\label{performance_analysis}

In this section, we analyze theoretically the asymptotic behavior of the proposed estimators. Since the parameters $|b|$ and $\sigma_o^2$ are not required for detection, our analysis focuses on the estimation of $a$. Regarding the DML estimator, the fact that we treat the phases $\phi_{bi}, i=1,\hd, N$ as deterministic unknowns means that the number of real unknown parameters for $N$ (complex) samples is $N+4$. Because the dimension of the parameter space grows linearly with the number of samples, the DML estimator falls within a special class of ML estimators that are based on ``partially-consistent observations''~\cite{neyman1948}. The estimated parameters can be classified into two groups. The first group are the incidental parameters $\phi_{bi}, i=1,\hd, N$. Each incidental parameter affects only a finite number of samples (a single sample in our case). The other group of parameters, $a$, $|b|$, and $\sigma_o^2$, affect all received samples and are thus called structural parameters. The estimation of structural parameters in the presence of incidental parameters has been investigated in the literature~\cite{neyman1948,mak1982,kay96,incidental2000} and is referred to as the Incidental Parameter Problem. It is well-known that the asymptotic properties of ML estimators, such as consistency, which hold when the dimension of the parameter space is fixed do not necessarily hold in the presence of incidental parameters~\cite{neyman1948}. It thus becomes important to investigate the asymptotic behavior of the DML estimator. We do this by explicitly taking into account the fact that $\phi_{1i}$ and $\phi_{2i}, i=1,\hd, N$ are equiprobably and independently chosen from the set $S_{M}$. Regarding the powers $P_1$, $P_2$  and $P_r$, the most common convention in TWRNs is to set $P_1=P_2=P_r$ (cf.~\cite{feifei_power,gao2009channel,cui2009memoryless}), even though some works use slightly different setups. For instance, the authors in~\cite{gao2009optimal} choose $P_2=2P_1$ and $P_r=\frac{1}{2}(P_1+P_2)$. In our analysis we consider the general case of $P_1=\alpha P_2$ and $P_r=\beta(P_1+P_2)$ for some $\alpha,\beta>0$.

The first asymptotic property we address is the consistency of the DML estimator, which we prove to hold when the parameter spaces of $a$ and $b$ are restricted to compact sets. This is true even for BPSK modulation. The second aspect we address is how the DML estimator behaves at high SNR (we will define SNR shortly). We show the estimator approaches the true channel with high probability at high SNR for $M>2$, while it encounters infinitely-many global minima at high SNR for $M=2$. 

We also analyze the high SNR behavior of the MCML estimator, showing that it approaches the true channel with high probability for $J=1$ (single pilot). For $J>1$, the pilots can be chosen such that the MCML estimator always approaches the true channel at high SNR.

\subsection{Consistency of the DML Estimator}
\label{sample_size_performance}
In this section, we study the behavior of the DML estimator as $N\rightarrow\infty$. We demonstrate that the DML estimator is consistent when the channel parameters $a$, $b$ belong to compact sets. This holds for any $M\geq2$ and for any $P_1, P_2, \sigma_r^2, \sigma_t^2$.

Before proceeding, we note that the estimator of $a$ in~\eqref{a_ML} belongs to the class of extremum estimators. An estimator $\hat{\omega}$ is called an extremum estimator (cf.~\cite{hayashi}, ~\cite{handbook_econometrics}) if there is an objective function $\Sigma_N(\omega)$ such that $\hat{\omega}=\arg\ \min\limits_{\omega\in\Omega}\Sigma_N(\omega)$, where $\Omega$ is the set of parameter values. The fundamental theorem for the consistency of extremum estimators can be summarized as follows:

\begin{thmx}
\label{extremum_theorem}
If $\omega$ belongs to a compact set $\Omega$ and $\Sigma_N(\omega)$ converges uniformly in probability to $\Sigma_o(\omega)$ as $N\rightarrow\infty$, where $\Sigma_o(\omega)$ is continuous and uniquely minimized at $\omega=\omega_o$, then $\hat{\omega}$ converges in probability to $\omega_o$. 
\end{thmx}

Thus, we need to establish that, as $N\rightarrow\infty$, the objective function $W_N(u)$ in~\eqref{random_sequence0} converges uniformly in probability to a deterministic function of $u$ which has a unique global minimum at $u=a$. Letting $v\triangleq a-u$, and $V_N(v)\triangleq W_N(a-v)$, we obtain
\begin{equation}
\label{random_sequence1}
V_N(v)=\frac{1}{N-1}\sum\limits_{i=1}^{N}\left(|y_i(v)|-\frac{1}{N}\sum\limits_{k=1}^{N}|y_k(v)|\right)^2,
\end{equation}
where $y_i(v)\triangleq \tilde{z}_i(a-v)$. The terms $y_i(v),\ i=1,\hd,N$ are independently-generated realizations of the RV $y(v)\triangleq \tilde{z}(a-v)$, and $V_N(v)$ is the sample variance of the envelope of $y(v)$. Let $\mathcal{V}(v)$ be the variance of $|y(v)|$, $\varphi_k\triangleq\frac{2\pi k}{M}$, and $\theta_k(v)\triangleq\phi_{v}-\phi_b+\varphi_k$ for $k=0,\hdots, M-1$. We show in Appendix~\ref{appendixD} that
\begin{equation}  
\label{var_exp}
\begin{split}
\mathcal{V}(v)&=A^2|v|^2P_1+A^2|b|^2P_2+\sigma_o^2-\left(\sqrt{\frac{\pi\sigma_o^2}{4M^2}}\sum_{k=0}^{M-1}L_{\text{\tiny$1/2$}}\left(-\lambda_{k}(v)\right)\right)^2,
\end{split}
\end{equation}
where $L_{\text{\tiny$1/2$}}(\cdot)$ is the Laguerre polynomial~\cite{abramowitz} with parameter $1/2$,
and 
\begin{equation}
\label{lambda_k_v}
\begin{split}
\lambda_{k}(v)&=\frac{1}{\sigma_o^2}\left(A^2|v|^2P_1+A^2|b|^2P_2+2A^2|v||b|\sqrt{P_1P_2}\cos(\theta_{k}(v))\right).
\end{split}
\end{equation}
The behavior of $\mathcal{V}(v)$ for $v=0$, which corresponds to $u=a$, is described in the following lemma:

\begin{lem}
\label{global_minimizer}
The variance $\mathcal{V}(v)$ of the random variable $|y(v)|$ has a unique global minimum occurring at $v=0$. 
\end{lem}
\vspace{-2ex}
\begin{proof}
See Appendix~\ref{appendixE}.
\end{proof}

To apply Theorem 1, it remains to establish that $V_N(v)$ converges uniformly in probability to $\mathcal{V}(v)$. Since $V_N(v)$ is the sample variance of $|y(v)|$ and $|y(v)|$ has a finite fourth central moment, $V_N(v)$ converges in probability to $\mathcal{V}(v)$~\cite{gracia}. The following lemma holds regarding uniform convergence in probability which is a stricter requirement than convergence in probability:
\begin{lem}
\label{uniform_convergence}
Suppose $a$, $b$ and $v$ all belong to compact sets, then $V_N(v)$ converges uniformly in probability to $\mathcal{V}(v)$. 
\end{lem}
\vspace{-2ex}
\begin{proof}
See Appendix~\ref{appendixF}.
\end{proof}
From Lemma~\ref{global_minimizer} and Lemma~\ref{uniform_convergence}, we see that all of the conditions of Theorem 1 are met if the channel parameters $a$ and $b$ belong to compact sets. Hence, the following theorem holds. 
\begin{thm}
\label{consistency_ML}
If the channel parameters $a$, $b$ belong to the compact sets $B_1$ and $B_2$, then the following estimator of $a$
\begin{equation*}
\label{a_ML_consistent}
\hat{a}=\arg\ \min_{u\in B_1}\Bigg\{\sum\limits_{i=1}^{N}\left(|z_i-Aut_{1i}|-\frac{1}{N}\sum\limits_{k=1}^{N}|z_k-Aut_{1k}|\right)^2\Bigg\}.
\end{equation*}
 is consistent.
\end{thm}

The compactness assumption for $a$ and $b$ can be satisfied by assuming that the magnitudes of $g_1$, $h_1$ and $h_2$ are bounded. If we treat $g_1$, $h_1$ and $h_2$ as complex Gaussian random variables, there is no upper bound on $|a|$ and $|b|$, strictly speaking, but we can always choose a sufficiently large $\xi$ such that $Prob(|a|, |b|\leq \xi)=1-\epsilon$, where $\epsilon$ can be made arbitrarily small.

\subsection{High SNR Behavior of the DML Estimator}
\label{SNRbehaviorDML}
We now investigate the behavior of the DML estimator at high transmit SNR. To simplify our presentation, we assume that $\sigma_r^2=\sigma_t^2=\sigma^2$. This assumption allows for a simple and intuitive definition of SNR as $\gamma\triangleq\frac{P_2}{\sigma^2}$. Let 
\begin{equation}
\label{defXv}
X(v)=\sum\limits_{i=1}^{N}\left(|Avt_{1i}+Abt_{2i}|-\frac{1}{N}\sum\limits_{k=1}^{N}|Avt_{1k}+Abt_{2k}|\right)^2
\end{equation}
be the objective funcion in~\eqref{a_ML} in the limit as $\sigma\rightarrow0$. The following lemma describes the behavior of the DML estimator as $\sigma\rightarrow0$.
\begin{lem}
\label{lemma_highSNR}
For fixed (finite) $N,$ the DML estimator approaches the true channel $a$ as $\sigma\rightarrow0$, except when the data symbols are such that the phase differences $\phi_{1i}-\phi_{2i},\ i=1,\hd, N$ take at most two distinct values, in which case the objective function encounters infinitely many global minima as $\sigma\rightarrow0$. Therefore, assuming $M$-PSK transmission, $\hat{a}\rightarrow a$ as $\sigma\rightarrow0$ with probability $P{\text{\tiny$M,N$}}=1-\left(\frac{2}{M}\right)^{N-1}(M-1)$. 
\end{lem}
\begin{proof}
See Appendix~\ref{appendixA}. 
\end{proof}

As we can see from Lemma~\ref{lemma_highSNR}, the DML estimator approaches the true channel except in the event that the phase differences $\phi_{1i}-\phi_{2i},\ i=1,\hd, N$ happen to take at most two distinct values. When this event occurs, the objective function of the estimator will have infinitely many global minima at high SNR. In fact, the occurence of this event also results in a singular Fisher information matrix as we shall see in Section~\ref{CRB_analysis}. For $M>2$, this event is highly unlikely for sufficiently large $N$. For instance, at $M=4$, the probability of this event is less than $5.6\times10^{-9}$ for $N\geq30$. Thus, the DML estimator approaches the true channel with high probability for $M>2$ as long as the sample size is not extremely small. In this case, the average MSE performance of the estimator keeps improving with SNR, and it can effectively achieve arbitrary accuracy for sufficiently high SNR. For $M=2$, however, $P{\text{\tiny$2,N$}}=0$ and the estimator always encounters infinitely-many global minima at high SNR because the difference $\phi_{1i}-\phi_{2i}$ cannot take more than two distinct values for $i=1,\hd, N$. In fact, for $M=2$ and for any non-zero $v, b,$ there are only two possible values that the terms $|vt_{1i}+bt_{2i}|, i=1,\hd,N$ can take, and they are $|\sqrt{P_1}v+\sqrt{P_2}b|$ and $|\sqrt{P_1}v-\sqrt{P_2}b|$. Whenever\footnote{The notation $x\perp y$ means that $\angle x-\angle y=\frac{\pi}{2}\ \mbox{mod}\ \pi$.} $v\perp b$, these two values are equal, which means that the terms $|vt_{1i}+bt_{2i}|$, $i=1,\hd,N$ are all equal regardless of the values of $t_{1i}$, $t_{2i}$, $i=1,\hd,N$. Therefore, all the values of $v$ such that $v\perp b$ are global minimizers of the objective function $X(v)$, and the estimator is not able to identify the true channel. Hence, the DML estimator performs poorly for $M=2$ as its MSE performance degrades at high SNR. Despite this behavior, the DML estimator is consistent for $M=2$ for any fixed $\sigma>0$, because for $\sigma>0$ only $v=0$ minimizes the variance $\mathcal{V}(v)$. This explains why, as we shall see, the DML estimator performs better at low SNR than high SNR for $M=2$. Since the performance of the estimator will also degrade for very low SNR, the MSE for $M=2$ exhibits a U-shaped behavior when plotted versus SNR. This is confirmed by our simulation results in Section~\ref{simulations}. 

We note that similar high SNR behavior is encountered if the noise variances $\sigma_r^2$ and $\sigma_t^2$ are assumed to be different (but fixed) and the transmit power $P_2$ becomes asymptotically high.

\subsection{High SNR Behavior of the MCML Estimator (M=2)}
\label{SNRbehaviorMCML}
We now investigate the high SNR behavior of the MCML estimator and show that by taking into account the BPSK structure of the data symbols it effectively avoids the problem of infinitely many global minima at high SNR. Let 
\begin{equation}
Y(v)\triangleq \sum\limits_{i=1}^{N}|Avt_{1i}+Abt_{2i}|^2-\frac{1}{N}\left(\sum\limits_{k=1}^{N}|Avt_{1k}+Abt_{2k}|\cdot |\cos(\Delta_k)|\right)^2
\end{equation}
be the objective function in~\eqref{CML8} as $\sigma\rightarrow0$, expressed in terms of $v=a-u$, where
\begin{equation}
\Delta_k=\angle(Avt_{1k}+Abt_{2k})-\angle(JAbP_2+Av\sum_{j=1}^{J}x_{1j}x_{2j}).
\end{equation}
Clearly, $Y(v)\geq0$ and $Y(0)=0$. It remains to investigate whether it is possible to have $Y(v)=0$ for $v\neq0$, i.e., whether there could be other global minima besides $v=0$. For this to happen, the terms $|Avt_{1k}+Abt_{2k}|,\ k=1,\hd, N$ have to be all equal and the terms $|\cos(\Delta_k)|$, $k=1,\hd, N$ have to be all equal to $1$. The first requirement can be met if $v\perp b$ or if the product $t_{1i}t_{2i}$ is the same for $i=1,\hd, N$. As for the second requirement, it is only satisfied when
\begin{equation}
\label{SNR_MCML}
\angle(Avt_{1k}+Abt_{2k})=\angle(JAbP_2+Av\sum_{j=1}^{J}x_{1j}x_{2j})\ \mbox{mod}\ \pi
\end{equation}
for $k=1, \hd, N$, where $\mbox{mod}$ is the modulo operator. Let $\mathcal{\mathcal{H}}$ be the event that the requirement in~\eqref{SNR_MCML} is satisified for $k=1,\hd, N$. For $J=1$, $\mathcal{\mathcal{H}}$ occurs with probability $P(\mathcal{\mathcal{H}})=1/2^{N}$. For $J>1$, however, the occurrence of $\mathcal{\mathcal{H}}$ can be completely avoided if the pilots are chosen such that the products $x_{1j}x_{2j},\ j=1,\hd, J$ are not all equal. If the pilots are randomly chosen, then $P(\mathcal{\mathcal{H}})=1/2^{N+J-1}$. Hence, depending on the choice of the pilots, the MCML estimator either always approaches the true channel or approaches it with a probability of $1-1/2^{N+J-1}$. Thus, the MCML estimator effectively resolves the problem of infinitely many global minima at high SNR even when a single pilot is used.

\subsection{MSE Performance of the GML Estimator}
\label{sample_average}

The estimator of $a$ in~\eqref{GMLa0} can be expanded as
\begin{equation}
\label{expanded_simple}
\begin{split}
\hat{a}_g&=a+\frac{b}{N\sqrt{\alpha}}\sum_{i=1}^{N}e^{\jmath(\phi_{2i}-\phi_{1i})}+\frac{1}{N\sqrt{\alpha P_2}}\sum_{i=1}^{N}e^{-j\phi_{1i}}h_2n_i+\frac{1}{NA\sqrt{\alpha P_2}}\sum_{i=1}^{N}e^{-j\phi_{1i}}\eta_{i}.
\end{split}
\end{equation}
It is straightforward to check that the estimator is unbiased. The resulting MSE is 
\begin{equation}
\label{MSE_MPSK}
\Expected\{|\hat{a}_g-a|^2\}=\frac{|b|^2}{N \alpha}+\frac{|h_2|^2\sigma^2}{N\alpha P_2}+\frac{\sigma^2}{NA^2\alpha P_2}.
\end{equation}
Since $\Expected\{|\hat{a}_g-a|^2\}\rightarrow0$ as $N\rightarrow\infty$, the estimator is consistent. Clearly, as $\sigma\rightarrow0$, the MSE performance of this estimator is limited by an error floor of $\frac{\textstyle |b|^2}{\textstyle N\alpha}$. 

\subsection{Cramer-Rao Bounds}
\label{CRB_analysis}

In this section, we obtain CRBs for the estimation of $a$ and $|b|$. The first bound is the standard CRB derived by treating the data symbols $t_{21}, \hd, t_{2N}$ as deterministic unknowns. We exclude the parameter $\sigma_o^2$ from our CRB derivation since its Fisher information is decoupled from the Fisher information of the other parameters. Let $\bsm{\theta}_R\triangleq[\Re\{a\},\Im\{a\}, |b|,\phi_{b1},\hd,\phi_{bN}]^{T}$ be the vector of real unknown parameters (excluding $\sigma_o^2$), and let $\bsm{I}(\bsm{\theta}_R)$ be the corresponding Fisher information matrix (FIM). The matrix $\bsm{I}(\bsm{\theta}_R)$ is given by
\begin{equation}
\bsm{I}(\bsm{\theta}_R)=\begin{bmatrix}
\ \bsm{H}_1&\bsm{H}_2\\
\ \bsm{H}_2^{T}&\bsm{H}_3
\end{bmatrix},
\end{equation}
where 
\begin{equation}
\bsm{H}_1=\frac{1}{\sigma_o^2}\begin{bmatrix}2A^2NP_1& 0&2A^2\Re\{e^{\jmath\phi_b}\bsm{t}_1^{H}\bsm{t}_2\}\\
0&2A^2NP_1&2A^2\Im\{e^{\jmath\phi_b}\bsm{t}_1^{H}\bsm{t}_2\}\\
2A^2\Re\{e^{\jmath\phi_b}\bsm{t}_1^{H}\bsm{t}_2\}&2A^2\Im\{e^{\jmath\phi_b}\bsm{t}_1^{H}\bsm{t}_2\}&2A^2NP_2
\end{bmatrix},
\end{equation}
\begin{equation}
\bsm{H}_2=\frac{1}{\sigma_o^2}\begin{bmatrix}
2A^2\Im\{b^*t_{11}t_{21}^{*}\}&\hd &2A^2\Im\{b^*t_{1N}t_{2N}^{*}\}\\
2A^2\Re\{b^*t_{11}t_{21}^{*}\}&\hd &2A^2\Re\{b^*t_{1N}t_{2N}^{*}\}\\
0&\hd&0
\end{bmatrix},
\end{equation}
and 
\begin{equation}
\bsm{H}_3=\frac{1}{\sigma_o^2}2A^2|b|^2P_2\bsm{I}_N.
\end{equation}
The resulting CRBs on the estimation of $a$ and $|b|$ are\footnote{The notation $[\bsm{A}]_{ij}$ is used to refer to the $(i,j)$th element of the matrix $\bsm{A}$.} 
\begin{equation}
CRB_a=[\bsm{I}^{-1}(\bsm{\theta}_R)]_{11}+[\bsm{I}^{-1}(\bsm{\theta}_R)]_{22},
\end{equation}
and
\begin{equation}
CRB_{|b|}=[\bsm{I}^{-1}(\bsm{\theta}_R)]_{33}.
\end{equation}
Moreover, using the Schur-complement property and letting $\tilde{\bsm{H}}\triangleq\bsm{H}_1-\bsm{H}_2\bsm{H}_3^{-1}\bsm{H}_2^{T}$,
we obtain\footnote{Although it is possible to evaluate the expressions in~\eqref{schur_a} and~\eqref{schur_b} and obtain closed-form expressions for $CRB_a$ amd $CRB_{|b|}$, the resulting expressions are quite lengthy, and are omitted for the sake of brevity.} 
\begin{equation}
\label{schur_a}
CRB_a=[\tilde{\bsm{H}}^{-1}]_{11}+[\tilde{\bsm{H}}^{-1}]_{22},
\end{equation}
and
\begin{equation}
\label{schur_b}
CRB_{|b|}=[\tilde{\bsm{H}}^{-1}]_{33}.
\end{equation}
The bounds $CRB_a$ and $CRB_{|b|}$ exist whenever $\tilde{\bsm{H}}$ is invertible. It can be shown that $\mbox{det}(\tilde{\bsm{H}})=0$ whenever the data symbols are such that the differences $\phi_{1i}-\phi_{2i},\ i=1,\hd, N$ take at most two distinct values, (recall that this is the same condition which results in infinitely many global minima at high SNR). Hence, the bounds do not exist for BPSK modulation since this condition is always met for $M=2$.

The bounds $CRB_a$ and $CRB_{|b|}$ are both functions of $\bsm{t}_1$ and $\bsm{t}_2$, and thus they hold for the particular realizations of $\bsm{t}_1$ and $\bsm{t}_2$ under consideration. In Section~\ref{simulations}, we use Monte-Carlo simulations to average $CRB_a$ and $CRB_{|b|}$ over many realizations of $\bsm{t}_1$ and $\bsm{t}_2$.

Another variant of the CRB, commonly used in the presence of random nuisance parameters is the Modified CRB (MCRB)~\cite{gini2000}. In deriving the MCRB, we compute a modified FIM (MFIM) by first extracting the submatrix of the conventional FIM which corresponds only to the nonrandom parameters ($\Re\{a\}$, $\Im\{a\}$ and $|b|$) and then obtaining the expectation of this submatrix with respect to the random nuisance parameters ($\bsm{t}_2$ in our case). Let $\bsm{\theta}'\triangleq[\Re\{a\}, \Im\{a\}, |b|]^{T}$, the MFIM is   
\begin{equation}
\label{MFIM}
\bsm{I}_m(\bsm{\theta}')\triangleq\Expected\{\bsm{H}_1\}=\frac{1}{\sigma_o^2}\begin{bmatrix}\ A^2NP_1& 0&0\\
0&2A^2NP_1&0\\
0&0&2A^2NP_2
\end{bmatrix}.
\end{equation}
Hence,
\begin{equation}
\label{MCRB_bounds}
MCRB_a=\frac{\sigma_o^2}{A^2NP_1}\ \ \ \mbox{and}\ \ \  MCRB_{|b|}=\frac{\sigma_o^2}{2A^2NP_2}.
\end{equation}
The bounds in~\eqref{MCRB_bounds} have the advantage of possessing a simple closed form. However, they are not as tight as the standard CRB and for this reason we will only use the bounds in~\eqref{schur_a} and~\eqref{schur_b} in our simulation results.

\section{Simulation Results}
\label{simulations}

In this section, we use Monte-Carlo simulations to numerically investigate the performance of the proposed algorithms. Our results are obtained assuming $P_r=P_1=P_2$, and they are averaged using a set of $300$ independent realizations of the channel parameters $h_1$, $h_2$, $g_1$ and $g_2$. These realizations are generated by treating $h_1$ and $h_2$ as correlated complex Gaussian random variables with mean zero, variance $1$ a correlation coefficient $\varrho=0.3$, and also treating $g_1$ and $g_2$ as correlated complex Gaussian random variables with the same mean, variance and correlation coefficient, but independent of $h_1$ and $h_2$. To generate correlated complex Gaussian random variables we follow the approach proposed in~\cite{tellambura_correlated}. The ML and MCML estimates are obtained using a two-dimensional grid-search with a step-size of $10^{-3}$. As before, the SNR is defined as $\gamma=\frac{P_2}{\sigma^2}$. Unless otherwise mentioned, MSE results are for the estimation of $a$.

We begin by comparing the MSE performance of the DML, GML and MCML estimators for $M=2$. We do not show the CRB in this case, since the FIM is singular. The MSE performance of the three estimators is plotted versus SNR for $N=45$ in Fig.~\ref{MSE_BPSK_SNR} and versus $N$ for an SNR of $20$ dB in Fig.~\ref{MSE_BPSK_N}. For the MCML estimator, $2$ pilots are employed to obtain an estimate of $\phi_b$. Both plots show that the DML estimator performs poorly for BPSK modulation and is outperformed by the GML and MCML estimators. Fig.~\ref{MSE_BPSK_SNR} shows the U-shaped behavior of the DML estimator for $M=2$ described in Section~\ref{SNRbehaviorDML}. The MCML estimator is superior to the GML estimator except at very low SNR. Moreover, as SNR increases, the MCML estimator improves steadily while the GML estimator hits an error floor. Fig.~\ref{MSE_BPSK_N} also demonstrates the superiority of the MCML estimator to the GML estimator. 

Next, we study the MSE performance of the DML and GML estimators for $M=4$ (QPSK modulation). The MSE performance of the two estimators is plotted versus SNR for $N=45$ in Fig.~\ref{MSE_QPSK_SNR} and versus $N$ for an SNR of $20$ dB in Fig.~\ref{MSE_QPSK_N}. The bound $CRB_a$ is shown as a reference in both plots. As we can see in Fig.~\ref{MSE_QPSK_SNR}, the DML estimator outperforms the GML estimator, except at low SNR. Moreover, the MSE performance of the DML estimator improves steadily with SNR and approaches $CRB_{a}$, while that of the GML estimator hits an error floor at high SNR. In Fig.~\ref{MSE_QPSK_SNR}, it is also noticed that the MSE of the GML estimator goes below $CRB_{a}$ at low SNR. This should not be a surprise since $CRB_{a}$ is derived by treating $\bsm{t}_2$ as deterministic, and the GML estimator is biased in this case. As we can see in Fig.~\ref{MSE_QPSK_N}, both estimators improve as $N$ increases, but the DML estimator is much closer to $CRB_{a}$. 

In Fig.~\ref{MSE_bQPSK_SNR}, the MSE performance of the DML and the GML estimators for the estimation of $|b|$ and the associated CRB are plotted versus SNR for $M=4$. For the GML estimator, we obtain an estimate of $|b|$ by substituting $\hat{a}_g$ in~\eqref{b_estimator}. Fig.~\ref{MSE_bQPSK_SNR} shows that the GML estimator is slightly better except at high SNR where it appears to hit an error floor and the DML estimator becomes better and approaches $CRB_{|b|}$. It must be noted that $|b|$ plays no role in detection when $M$-PSK modulation is used.

Our next goal is to compare the SER performance of the DML estimator with that of the training-based LS estimator in order to investigate the tradeoffs between accuracy and spectral efficiency that result from following the blind approach. We focus on small sample sizes since this is more suitable for modern-day cellular systems. We note that, when channels are nonreciprocal, the training-based LS estimator is an efficient estimator which coincides with the training-based ML estimator. As a reference, we also plot the SER performance under perfect channel knowledge. The phase $\phi_b$ is estimated blindly using the Viterbi-Viterbi algorithm\footnote{Using the Viterbi-Viterbi algorithm, we estimate $\phi_b$ by $\hat{\phi}_b\triangleq\frac{1}{M}\angle \sum_{i=1}^{N} |\tilde{z}_i(\hat{a})|^2e^{\jmath M\angle\tilde{z}_i(\hat{a})+\pi},$ where $\tilde{z}_i(\hat{a}), i=1,\hd, N$ are the resulting $N$ samples after the estimate $\hat{a}$ is used to remove self-interference.}, and a small number of pilots is employed to resolve the resulting $M$-fold ambiguity using the unique word method~\cite{synchronization}. In Fig.~\ref{SER_QPSK1}, we show the SER performance of the two algorithms assuming that the channel is fixed for the duration of 20 samples. For the DML estimator, 2 pilots and 18 data symbols are transmitted. The 2 pilots are used to resolve the $M$-fold ambiguity, and all 20 samples are used to blindly estimate $a$. For the LS estimator, we estimate $a$ and $b$ using 4 pilot symbols\footnote{For the training-based LS algorithm, the pilots are chosen such that the pilot vectors from the two terminals are orthogonal to each other.} and we transmit 16 data symbols. As we can see from Fig.~\ref{SER_QPSK1}, the SER performance of the DML estimator is very close to that of the LS estimator (approximately 0.6 dB away). In Fig.~\ref{SER_QPSK2}, we assume that the channel is fixed for the duration of 40 samples. In the DML case 4 pilots and 36 data symbols are transmitted, and in the LS case 8 pilots and 32 data symbols are transmitted. The performance of DML estimator is again very close to that of the LS estimator (approximately 0.4 dB away), and it is only $1.5$ dB away from the performance under perfect CSI. In both examples, for the DML algorithm we use $90\%$ of the channel coherence time to transmit data and $10\%$ to transmit pilots, while for the LS algorithm we are use $80\%$ of the coherence time to transmit data and $20\%$ to transmit pilots, which demonstrates that the DML estimator offers a better tradeoff between accuracy and spectral efficiency.

It is worth noting that the channel coherence times considered in our SER studies are compatible with current wideband cellular systems. For instance, considering a wideband system with $256$ carriers and a sampling frequency of $10^{-6}$s, the channel has to be invariant for $5.12$ms in the first example and for $10.24$ms in the second example. Assuming a central frequency of $2.4$GHz and a vehicle speed of $10$m/s, the coherence time is around $12.5$ms, which is larger than the required coherence time in both examples.   

\section{Conclusions}
\label{conclusions}

In this work, we proposed the DML algorithm for blind channel estimation in AF TWRNs employing $M$-PSK modulation. The DML estimator was derived by treating the data symbols as deterministic unknowns. For comparison, we derived the GML estimator by treating the data symbols as Gaussian-distributed nuisance parameters. We showed that the DML estimator is consistent. For $M>2$, we showed that it approaches the true channel with high probability as SNR increases and that its MSE performance is superior to that of the GML estimator for medium-to-high SNR and approaches the derived CRB at high SNR. In contrast, the GML estimator suffers from an error floor at high SNR. We also compared the SER performance of the DML estimator with that of the training-based LS estimator and showed that the DML approach provides a better tradeoff between accuracy and spectral efficiency. For $M=2$, however, we showed that the DML estimator performs poorly and proposed the MCML estimator which explicitly takes into account the structure of the BPSK signal. This estimator outperforms the GML estimator except at very low SNR and approaches the true channel at high SNR.

\subsection*{Acknowledgments}
The authors of this work would like to thank Mr. Ali Shahrad for his help in proving Theorem 2 and the Associate Editor as well as the anonymous reviewers for their valuable comments and suggestions on the original manuscript.

\appendices

\section{}
\label{appendixD}
In this appendix, we derive the closed-form expression for the variance $\mathcal{V}(v)$ of the RV $|y(v)|$. We recall that
\begin{equation}
\label{eqappD1}
y(v)=Av\sqrt{P_1}e^{\jmath\phi_1}+Ab\sqrt{P_2}e^{\jmath\phi_2}+Ah_2n+\eta.
\end{equation}
The variance of $|y(v)|$ is $\mathcal{V}(v)=\Expected\{|y(v)|^2\}-\Expected\{|y(v)|\}^2$. It can be easily shown that
\begin{equation}
\label{first_part}
\Expected\{|y(v)|^2\}=A^2|v|^2P_1+A^2|b|^2P_2+\sigma_o^2.
\end{equation}
To find a closed form for $\Expected\left\{|y(v)|\right\}$, we first obtain the conditional expectation $\Expected\left\{|y(v)|\ \left|\ \phi_1,\phi_2\right.\right\}$. We can see from~\eqref{eqappD1} that $\Re\{y(v)\}$ and $\Im\{y(v)\}$ conditioned on $\phi_1$ and $\phi_2$, are Gaussian-distributed with conditional means
\begin{equation*}
\label{real_condtional_mean}
\begin{split}
\Expected\{\Re(y(v))\ |\ \phi_1,\phi_2\}&=A|v|\sqrt{P_1}\cos(\phi_{v}+\phi_1)+A|b|\sqrt{P_2}\cos(\phi_b+\phi_2),
\end{split}
\end{equation*}
and
\begin{equation*}
\label{real_condtional_mean}
\begin{split}
\Expected\{\Im(y(v))\ |\ \phi_1,\phi_2\}=A|v|\sqrt{P_1}\sin(\phi_{v}+\phi_1)+A|b|\sqrt{P_2}\sin(\phi_b+\phi_2)
\end{split}
\end{equation*}
and a conditional variance of $\sigma_o^2/2$. Hence, when conditioned on $\phi_1$ and $\phi_2$, $|y(v)|$ is a noncentral Chi random variable whose mean is given by~\cite{oberto2006}
\begin{equation*}
\label{chi_mean}
\Expected\left\{|y(v)|\ \left|\ \phi_1,\phi_2\right.\right\}=\sqrt{\frac{\pi\sigma_o^2}{4}}L_{\text{\tiny$1/2$}}\left(-\lambda(v;\phi_1,\phi_2)\right),
\end{equation*}
where 
\begin{equation}
\label{first_lambda}
\begin{split}
\lambda(v;\phi_1,\phi_2)&\triangleq\frac{1}{\sigma_o^2}\left(\Expected\{\Re(y(v))\ |\phi_1,\phi_2\}^2+\Expected\{\Im(y(v))\ |\phi_1,\phi_2\}^2\right)\\
&=\frac{1}{\sigma_o^2}\left(A^2|v|^2P_1+A^2|b|^2P_2+2A^2|v||b|\sqrt{P_1P_2}\cos(\phi_{v}-\phi_b+\phi_1-\phi_2)\right)
\end{split}
\end{equation}
and $L_{\text{\tiny$1/2$}}(x)=e^{x/2}\left[(1-x)I_0(x/2)+xI_1(x/2)\right]$ is the Laguerre polynomial with parameter $1/2$, and $I_{\tau}(\cdot)$ is the Modified Bessel Function of the First Kind of order $\tau$~\cite{abramowitz}. The last equation shows that $\lambda(v;\phi_1,\phi_2)$ depends on the difference $(\phi_1-\phi_2)\ \mbox{mod}\ 2\pi$, i.e., $\lambda(v;\phi_1,\phi_2)=\lambda(v;\phi_1-\phi_2)$. The difference, $(\phi_1-\phi_2)\ \mbox{mod}\ 2\pi$ takes the values $\varphi_k=\frac{2\pi k}{M},\ k=0,\hdots, M-1$ with equal probability. Let $\lambda_k(v)\triangleq \lambda(v;\phi_1-\phi_2=\varphi_k)$, the unconditional mean $\Expected\left\{|y(v)|\right\}$ is given by
\begin{equation}
\label{unconditional_mean}
\Expected\{|y(v)|\}=\sum_{k=0}^{M-1}\sqrt{\frac{\pi\sigma_o^2}{4M^2}}L_{\text{\tiny$1/2$}}\left(-\lambda_k(v)\right),
\end{equation}
which completes the proof of~\eqref{var_exp}. 

\section{}
\label{appendixE}
In this appendix, we show that $\mathcal{V}(v)$ has a unique global minimum at $v=0$. Let $\nu\triangleq \sqrt{\frac{1}{\sigma_o^2}\left(A^2|v|^2P_1\right)}$, $\zeta\triangleq\sqrt{\frac{1}{\sigma_o^2}\left(A^2|b|^2P_2\right)}$, and let
\begin{equation}
\label{appE3}
D(\nu)\triangleq\nu^2+\zeta^2-\frac{\pi}{4M^2}\left(\sum\limits_{k=0}^{M-1}L_{\text{\tiny$1/2$}}\left(-\left[\nu^2+\zeta^2+2\nu\zeta\cos(\theta_k(v))\right]\right)\right)^2.
\end{equation}
It is sufficient to prove that $D(\nu)$ has a unique global minimum at $\nu=0$. It is straightforward to verify that $L_{\text{\tiny$1/2$}}(-x)$ is a strictly concave function. Therefore,
\begin{equation} 
\label{concave}
\begin{split}
\frac{1}{M}\sum\limits_{k=0}^{M-1}L_{\text{\tiny$1/2$}}\left(-\left[\nu^2+\zeta^2+2\nu\zeta\cos(\theta_k(v))\right]\right)
&\leq L_{\text{\tiny$1/2$}}\left(-\frac{1}{M}\sum\limits_{k=0}^{M-1}\left[\nu^2+\zeta^2+2\nu\zeta\cos(\theta_k(v))\right]\right)\\
&=L_{\text{\tiny$1/2$}}\left(-\left[\nu^2+\zeta^2\right]\right),
\end{split}
\end{equation}
where we have used the fact that $\sum\limits_{k=0}^{M-1}\cos(\phi_v-\phi_b+\frac{2k\pi}{M})=0$. Therefore,
\begin{equation}
\label{appE4}
D(\nu)\geq \nu^2+\zeta^2-\frac{\pi}{4}\left(L_{\text{\tiny$1/2$}}\left(-\left[\nu^2+\zeta^2\right]\right)\right)^2.
\end{equation}
For $\zeta\neq0$, the equality holds if and only if $\nu=0$. Let $F(\nu)\triangleq \nu^2+\zeta^2-\frac{\pi}{4}\left(L_{\text{\tiny$1/2$}}\left(-\left[\nu^2+\zeta^2\right]\right)\right)^2$, it is sufficient to show that $F(\nu)$ has a unique global minimum at $\nu=0$. We have
\begin{equation}
\label{appE5}
\frac{d}{d\nu}F(\nu)=2\nu-\nu\pi L_{\text{\tiny$1/2$}}\left(-\left[\nu^2+\zeta^2\right]\right)\mathcal{P}(\nu^2+\zeta^2),
\end{equation}
where $\mathcal{P}(x)\triangleq\frac{1}{2}e^{-x/2}\left(I_0(x/2)+I_1(x/2)\right)$. Since $\frac{d}{d\nu}F(\nu)=0$ for $\nu=0$, it is sufficient to show that $\frac{dF(\nu)}{d\nu}>0$ for $\nu\neq0$. In other words, we have to show that $\frac{\pi}{2}L_{\text{\tiny$1/2$}}\left(-\left[\nu^2+\zeta^2\right]\right)\mathcal{P}(\nu^2+\zeta^2)<1$. Let $\rho\triangleq\nu^2+\zeta^2$ and let
\begin{equation}
\label{appE7}
\begin{split}
Q(\rho)&\triangleq \frac{\pi}{2}L_{\text{\tiny$1/2$}}\left(-\rho\right)\mathcal{P}(\rho)\\
&=\frac{\pi}{4}e^{-\rho}\left[(1+\rho)I_0\left(\frac{\rho}{2}\right)+\rho I_1\left(\frac{\rho}{2}\right)\right]\left[I_0\left(\frac{\rho}{2}\right)+I_1\left(\frac{\rho}{2}\right)\right],
\end{split}
\end{equation}
we have to show that $Q(\rho)<1$ for $\rho>0$. We will do this by showing that $\frac{d}{d\rho}Q(\rho)>0$ and that $\lim\limits_{\rho\rightarrow\infty}Q(\rho)=1$. We have
\begin{equation}
\label{appE8}
\frac{d}{d\rho}Q(\rho)=\frac{\pi}{2}e^{-\rho}\left[\frac{1}{4}I_0\left(\frac{\rho}{2}\right)^2-\frac{1}{4}I_1\left(\frac{\rho}{2}\right)^2-\frac{1}{2\rho}I_0\left(\frac{\rho}{2}\right)I_1\left(\frac{\rho}{2}\right)\right].
\end{equation}
Let $U(\rho)\triangleq \rho I_0\left(\frac{\rho}{2}\right)^2-\rho I_1\left(\frac{\rho}{2}\right)^2-2I_0\left(\frac{\rho}{2}\right)I_1\left(\frac{\rho}{2}\right)$. Hence $\frac{d}{d\rho}Q(\rho)=\frac{\pi}{8\rho}e^{-\rho}U(\rho)$. Moreover, $\frac{d}{d\rho}U(\rho)=\frac{2}{\rho}I_0\left(\frac{\rho}{2}\right)I_1\left(\frac{\rho}{2}\right)>0$. Thus, $U(\rho)$ is strictly increasing for $\rho>0$. Since $U(0)=0$, we have that $U(\rho)>0$ for $\rho>0$, which implies that $\frac{d}{d\rho}Q(\rho)>0$ for $\rho>0$. 

It remains to show that $\lim\limits_{\rho\rightarrow\infty}Q(\rho)=1$. For large arguments, the $I_{\tau}(\cdot)$ has the following asymptotic expansion~\cite{abramowitz}
\begin{equation}
\label{appE10}
I_{\tau}(x)\approx\frac{e^{x}}{\sqrt{2\pi x}}\left\{1-\frac{\varepsilon-1}{8x}+\frac{(\varepsilon-1)(\varepsilon-9)}{2!(8x)^2}-\hdots\right\},
\end{equation}
where $\varepsilon=4\tau^2$. We rewrite $Q(\rho)$ as
\begin{equation}
Q(\rho)=\frac{\pi}{4}e^{-\rho}\left[(1+\rho)I_0\left(\frac{\rho}{2}\right)^2+(1+2\rho)I_0\left(\frac{\rho}{2}\right)I_1\left(\frac{\rho}{2}\right)+\rho I_1\left(\frac{\rho}{2}\right)^2\right].
\end{equation}
Using the expansion in~\eqref{appE10}, we obtain
\begin{equation}
\label{appE11}
\lim\limits_{\rho\rightarrow\infty}\frac{\pi}{4}e^{-\rho}(1+\rho)I_0\left(\frac{\rho}{2}\right)^2=\lim\limits_{\rho\rightarrow\infty}\frac{\pi}{4}e^{-\rho}\rho I_1\left(\frac{\rho}{2}\right)^2=\frac{1}{4},
\end{equation}
and
\begin{equation}
\label{appE12}
\lim\limits_{\rho\rightarrow\infty}\frac{\pi}{4}e^{-\rho}(1+2\rho)I_0\left(\frac{\rho}{2}\right)I_1\left(\frac{\rho}{2}\right)=\frac{1}{2}.
\end{equation}
Therefore, $\lim\limits_{\rho\rightarrow\infty}Q(\rho)=1$, which completes the proof. 

\section{}
\label{appendixF}
In this appendix, we prove that $V_N(v)$ converges uniformly in probability to $\mathcal{V}(v)$ when $a$, $b$ and $v$ belong to compact sets. Suppose $v\in\mathcal{C}$. By definition, $V_N(v)$ converges uniformly in probability to $\mathcal{V}(v)$ when $\sup_{v\in \mathcal{C}}\big|V_N(v)-\mathcal{V}(v)\big|$ converges in probability to zero as $N\rightarrow\infty$. Since $y_i(v)=\tilde{z}_i(a-v)=z_i-A(a-v)t_{1i},$ $V_N(v)$ is a function of the parameter $v$, the observations $\bsm{z}$ and known data symbols $\bsm{t}_1$. According to Lemma 2.9 in [28, Ch. 36], a sufficient condition for uniform convergence in probability when $\mathcal{C}$ is compact is the existence of a function $\mathcal{F}_N({\boldsymbol z},{\boldsymbol t_1})$ with bounded expectation $\Expected\{\mathcal{F}_N({\boldsymbol z},{\boldsymbol t_1})\}$ such that for all $v_1, v_2\in \mathcal{C}$, $\big|V_N(v_1)-V_N(v_2)\big|\leq \mathcal{F}_N({\boldsymbol z},{\boldsymbol t_1})\big|v_1-v_2\big|$. 

Using the triangular inequality, we obtain
\begin{equation}
\label{appF1}
\begin{split}
\big|V_N(v_1)-V_N(v_2)\big|&\leq\frac{1}{N-1}\sum\limits_{i=1}^{N}\left|\left(|y_i(v_1)|-\frac{1}{N}\sum\limits_{k=1}^{N}|y_k(v_1)|\right)^2-\left(|y_i(v_2)|-\frac{1}{N}\sum\limits_{k=1}^{N}|y_k(v_2)|\right)^2\right|\\
&=\frac{1}{N-1}\sum\limits_{i=1}^{N}\Bigg(\left||y_i(v_1)|-|y_i(v_2)|-\frac{1}{N}\sum\limits_{k=1}^{N}|y_k(v_1)|+\frac{1}{N}\sum\limits_{k=1}^{N}|y_k(v_2)|\right|\\
&\ \ \ \ \ \ \ \ \ \ \ \ \ \times\left||y_i(v_1)|+|y_i(v_2)|-\frac{1}{N}\sum\limits_{k=1}^{N}|y_k(v_1)|-\frac{1}{N}\sum\limits_{k=1}^{N}|y_k(v_2)|\right|\Bigg).
\end{split}
\end{equation}
For $i=1,\hd, N,$ let
\begin{equation}
\label{termA}
\Upsilon_i(v_1,v_2)\triangleq \left||y_i(v_1)|-|y_i(v_2)|-\frac{1}{N}\sum\limits_{k=1}^{N}|y_k(v_1)|+\frac{1}{N}\sum\limits_{k=1}^{N}|y_k(v_2)|\right|,
\end{equation}
and
\begin{equation}
\label{termB}
\Lambda_i(v_1,v_2)\triangleq \left||y_i(v_1)|+|y_i(v_2)|-\frac{1}{N}\sum\limits_{k=1}^{N}|y_k(v_1)|-\frac{1}{N}\sum\limits_{k=1}^{N}|y_k(v_2)|\right|. 
\end{equation}
We can use the triangular inequality to get
\begin{equation}
\label{appF2}
\begin{split}
\Upsilon_i(v_1,v_2)&\leq |y_i(v_1)-y_i(v_2)|+\frac{1}{N}\sum\limits_{k=1}^{N}\left|y_k(v_1)-y_k(v_2)\right|\\
&=2AP_1|v_1-v_2|,
\end{split}
\end{equation}
and
\begin{equation}
\label{appF3}
\begin{split}
\Lambda_i(v_1,v_2)&\leq |y_i(v_1)|+|y_i(v_2)| +\frac{1}{N}\sum\limits_{k=1}^{N}|y_k(v_1)|+\frac{1}{N}\sum\limits_{k=1}^{N}|y_k(v_2)|.
\end{split}
\end{equation}
Noting that $|y_i(v)|\leq |z_i|+A|a|P_1+AP_1|v|$, we further obtain
\begin{equation}
\Lambda_i(v_1,v_2)\leq 2|z_i|+ \frac{2}{N}\sum\limits_{k=1}^{N}|z_k|+4AP_1|a|+2AP_1(|v_1|+|v_2|).
\end{equation}
Since the set $\mathcal{C}$ is compact, there exists $\mathcal{U}>0$ such that $|v|\leq\mathcal{U}, \forall v\in \mathcal{C}$.
Hence, we obtain the following upper bound on $\Lambda_i(v_1, v_2)$:
\begin{equation}
\label{appF4}
\Lambda_i(v_1,v_2)\leq2|z_i|+ \frac{2}{N}\sum\limits_{k=1}^{N}|z_k|+2AP_1|a|+4A\mathcal{U}.
\end{equation}
Combining~\eqref{appF2}, and~\eqref{appF4}, we obtain
\begin{equation}
\label{appF4_1}
\begin{split}
\big|V_N(v_1)-V_N(v_2)\big|&\leq\frac{8AP_1}{N-1}\left[ANP_1\mathcal{U}+ANP_1|a|+\sum\limits_{i=1}^{N}|z_i|\right]|v_1-v_2|.
\end{split}
\end{equation}
Noting that $z_i=y_i(a)$, we have from~\eqref{unconditional_mean} that
\begin{equation}
\begin{split}
\Expected\{|z_i|\}&\leq \sqrt{\frac{\pi\sigma_o^2}{4}} L_{\text{\tiny$1/2$}}\left(-\frac{1}{\sigma_o^2}\left(A|a|P_1+A|b|P_2\right)^2\right).
\end{split}
\end{equation}
Since both $a$, and $b$ belong to compact sets, there exists $\mathcal{M}_1>0$ and $\mathcal{M}_2>0$ such that $|a|\leq\mathcal{M}_1$ and  $|b|\leq\mathcal{M}_2$. Hence,
\begin{equation}
\Expected\{|z_i|\}\leq\sqrt{\frac{\pi\sigma_o^2}{4}} L_{\text{\tiny$1/2$}}\left(-\frac{1}{\sigma_o^2}\left(A\mathcal{M}_1P_1+A\mathcal{M}_2P_2\right)^2\right).
\end{equation}
Letting, $\mathcal{G}\triangleq \sqrt{\frac{\pi\sigma_o^2}{4}} L_{\text{\tiny$1/2$}}\left(-\frac{1}{\sigma_o^2}\left(A\mathcal{M}_1P_1+A\mathcal{M}_2P_2\right)^2\right)$, we obtain
\begin{equation}
\label{appF5}
\Expected\left\{\frac{8AP_1}{N-1}\left[ANP_1\mathcal{U}+ANP_1|a|+\sum\limits_{i=1}^{N}|z_i|\right]\right\}\leq16A^2(\mathcal{U}+\mathcal{M}_1)+16A\mathcal{G},
\end{equation}
which completes the proof.

\section{}
\label{appendixA}
In this appendix, we prove Lemma~\ref{lemma_highSNR}. It is obvious from~\eqref{defXv} that $X(v)\geq0$ with equality if and only if the terms $|Avt_{1i}+Abt_{2i}|,\ i=1,\hd, N$ are all equal. Because of the constant modulus nature of the data symbols, this occurs at $v=0$ for any $M$, which means that there is always a global minimum at $v=0$ (i.e., at $u=a$). To see whether we also have $X(v)=0$ for some $v\neq0$, we rewrite $|Avt_{1i}+Abt_{2i}|$ as
\begin{equation}
\label{highSNR1}
\begin{split}
|Avt_{1i}+Abt_{2i}|&=|A|v|\sqrt{P_1}e^{\jmath\phi_v+\phi_{1i}}+A|b|\sqrt{P_2}e^{\jmath{\phi_b+\phi_{2i}}}|\\
&=|A|v|\sqrt{P_1}e^{\jmath\phi_v-\phi_b+\phi_{1i}-\phi_{2i}}+A|b|\sqrt{P_2}|.
\end{split}
\end{equation}
Let $\chi_i\triangleq \cos(\phi_v-\phi_b+\phi_{1i}-\phi_{2i}),\ i=1,\hd, N$. It is clear from~\eqref{highSNR1} that $X(v)$ is zero whenever the terms $\chi_i,\ i=1,\hdots,N$ are all equal. We denote by $\mathcal{E}$ the event that the terms $\chi_i,\ i=1,\hdots,N$ are equal. Let $\psi_i\triangleq \phi_{1i}-\phi_{2i}$, then the values $\psi_i$, $i=1,\hdots,N$ are i.i.d. realizations of the discrete random variable $\Psi$ which takes values from the set $S_{\Psi}=\left\{\frac{2\ell\pi}{M}, \ell=0,\hdots,M-1\right\}$ with equal probability. We will now show that the event $\mathcal{E}$ occurs if and only if the values $\psi_i,\ i=1,\hdots,N$ are chosen from the same size $2$ subset of $S_{\Psi}$. Suppose that $\psi_1$ is fixed, and that we are choosing the remaining phases such that $\cos(\psi_i+\vartheta)=\cos(\psi_1+\vartheta)$ for $i=2,\hdots,N$. If $\psi_{\kappa}$ is different from $\psi_1$ for some index $\kappa$, then  $\cos(\psi_{\kappa}+\phi_v-\phi_b)=\cos(\psi_1+\phi_v-\phi_b)$ can only be satisfied if $2(\phi_v-\phi_b)=(-\psi_1-\psi_{\kappa})$. For the remaining phases with indices $i=2,\hdots,N,\ i\neq\kappa$, the equality $\cos(\psi_i+\phi_v-\phi_b)=\cos(\psi_1+\phi_v-\phi_b)$ holds only if $\psi_i=\psi_1$ or $\psi_i=\psi_{\kappa}$. Therefore, the terms $\chi_i,\ i=1,\hdots,N$ are equal if and only if the phases $\psi_i,\ i=1,\hdots,N$ take at most two distinct values, i.e, the probability that $\mathcal{E}$ occurs is
\begin{equation}
\label{appC1}
P\left(\mathcal{E}\right)= {M\choose2} \frac{2^N}{M^N}=\left(\frac{2}{M}\right)^{N-1}(M-1),
\end{equation}
where $M\choose2$ is the number of distinct subsets of size $2$ that can be chosen from a set of size $M$. Suppose that $\mathcal{E}$ occurs and that $\psi_{\text{\tiny$\Sigma$}}$ is the sum of the two distinct phase values, then $X(v)=0$ for all the values of $v$ that satisfy $2(\phi_v-\phi_b)=-\psi_{\text{\tiny$\Sigma$}}$, which means that there are infinitely many global minimizers of $X(v)$. Hence, the probability that $X(v)$ has a unique minimum at $v=0$ is 
$P{\text{\tiny$M,N$}}=1-\left(\frac{2}{M}\right)^{N-1}(M-1)$.
\bibliographystyle{IEEEtran}
\bibliography{IEEEabrv,twrn_bib}

\begin{figure}[!ht]
\centering
\includegraphics[width=4.4in, height=3.3in]{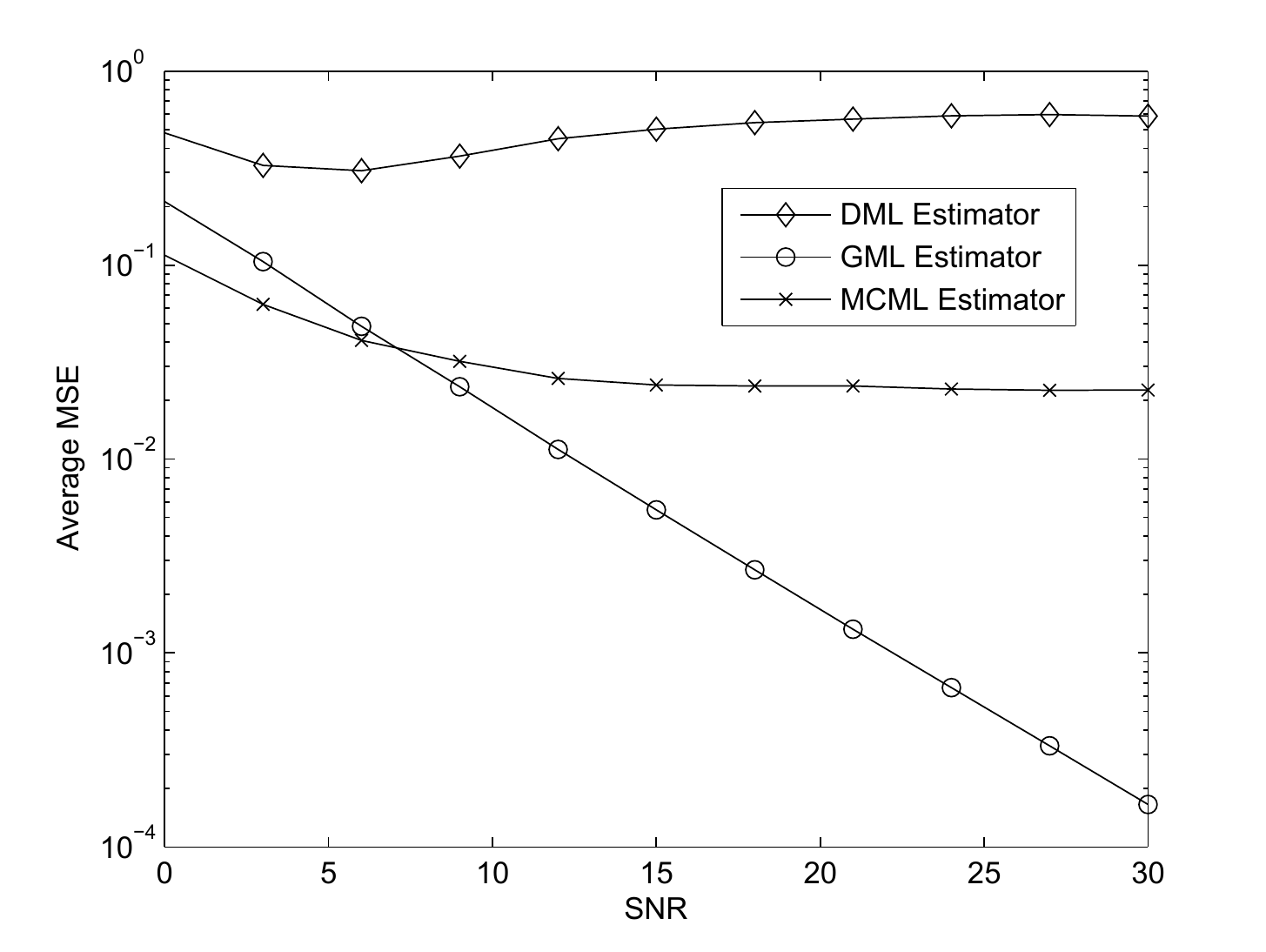}
\caption{Average MSE of the DML, GML and MCML algorithms for the estimation of $a$, plotted versus SNR for $M=2$ (BPSK modulation) and $N=45$. For the MCML algorithm, we use two pilots to estimate $\phi_b$.}
\label{MSE_BPSK_SNR}
\end{figure}

 \begin{figure}[!ht]
\centering
\includegraphics[width=4.4in, height=3.3in]{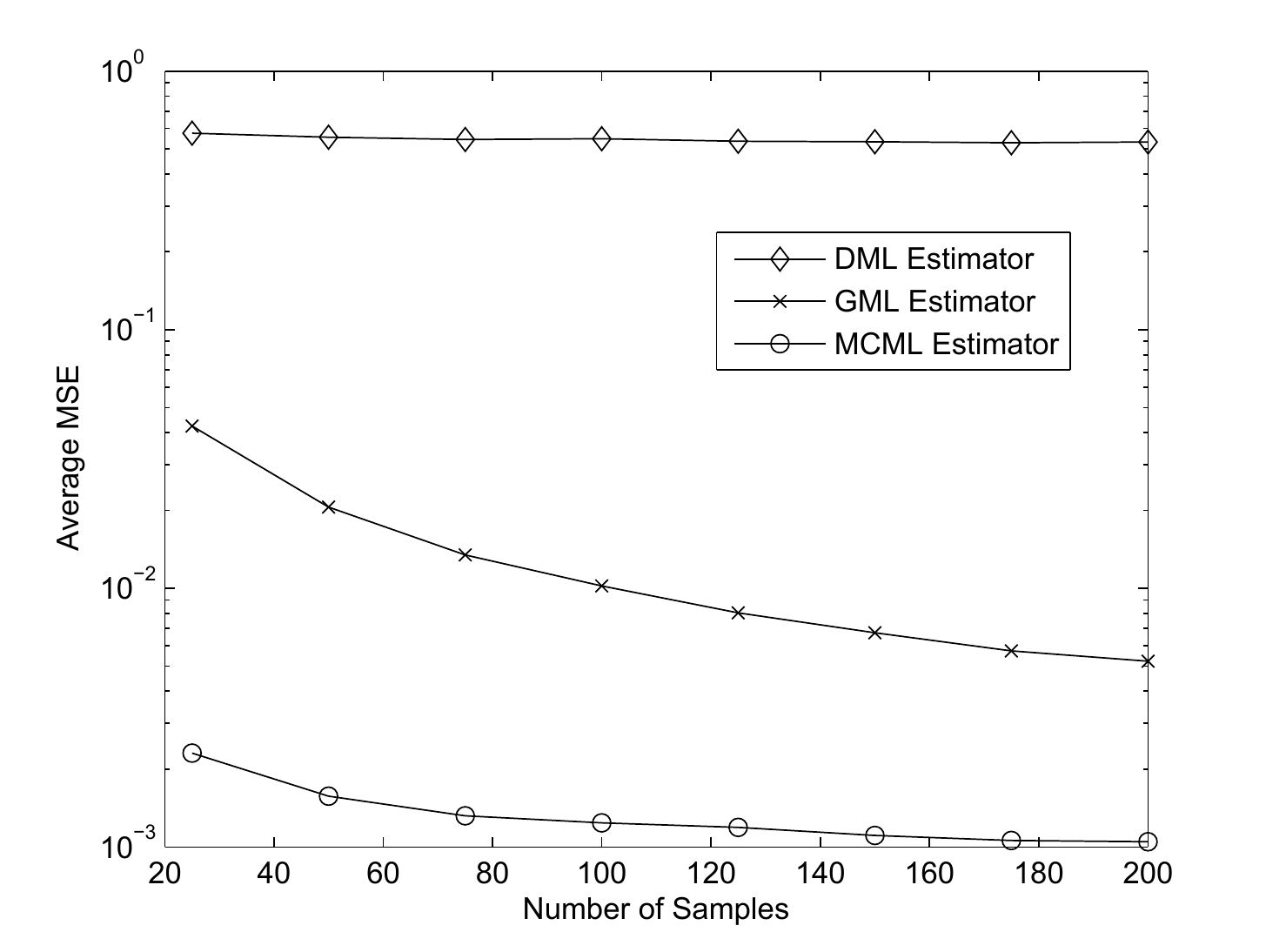}
\caption{Average MSE of the DML, GML and MCML algorithms for estimation of $a$, plotted versus $N$ for $M=2$ and an SNR of $20$ dB. For the MCML algorithm, we use two pilots to estimate $\phi_b$.}
\label{MSE_BPSK_N}
\end{figure}

 \begin{figure}[!ht]
\centering
\includegraphics[width=4.4in, height=3.3in]{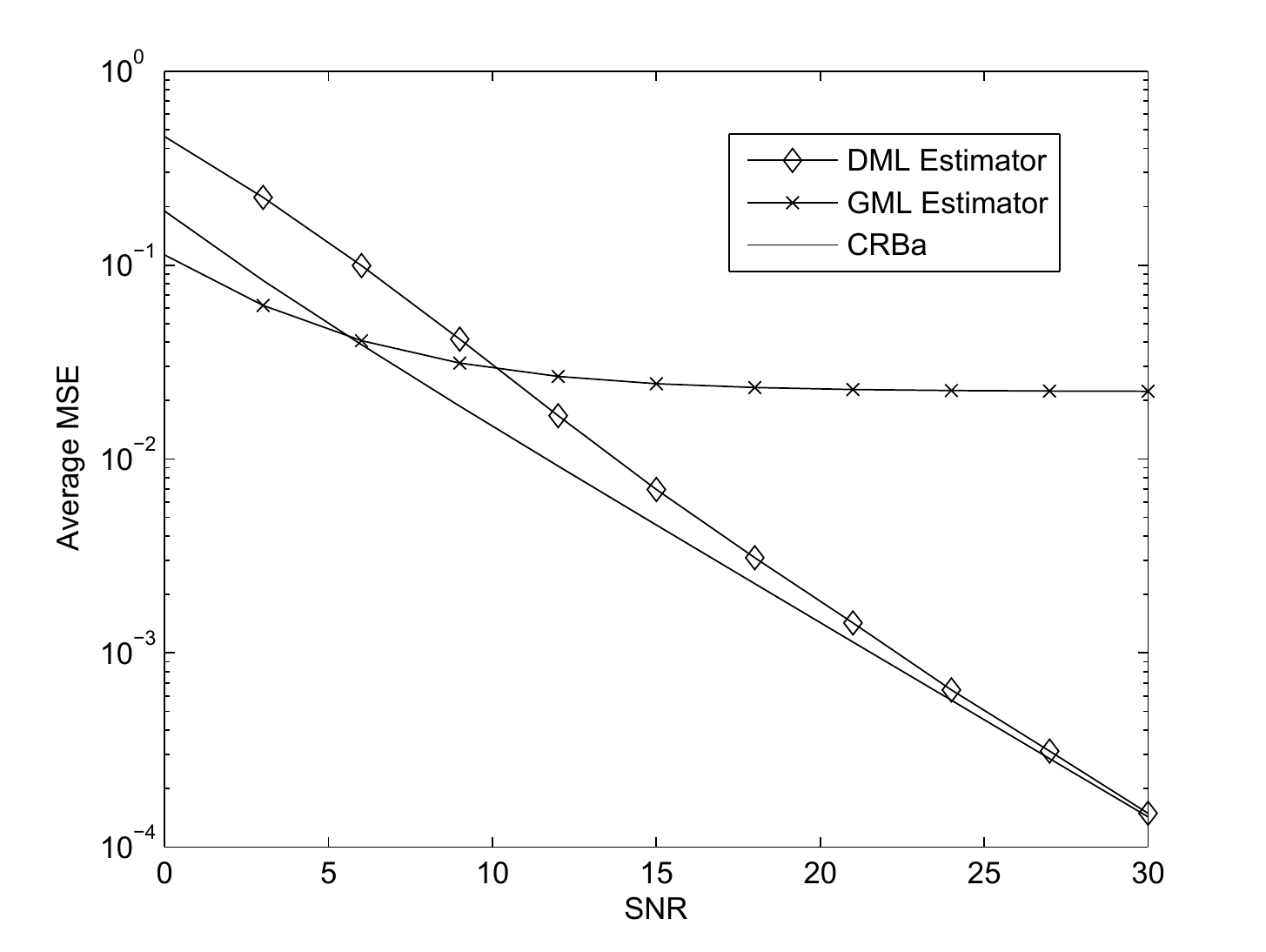} 
\caption{Average MSE of the DML and GML algorithms for the estimation of $a$ and the bound $CRBa$ plotted versus SNR for $M=4$ (QPSK modulation) and $N=45$.}
\label{MSE_QPSK_SNR}
\end{figure}

 \begin{figure}[!ht]
\centering
\includegraphics[width=4.4in, height=3.3in]{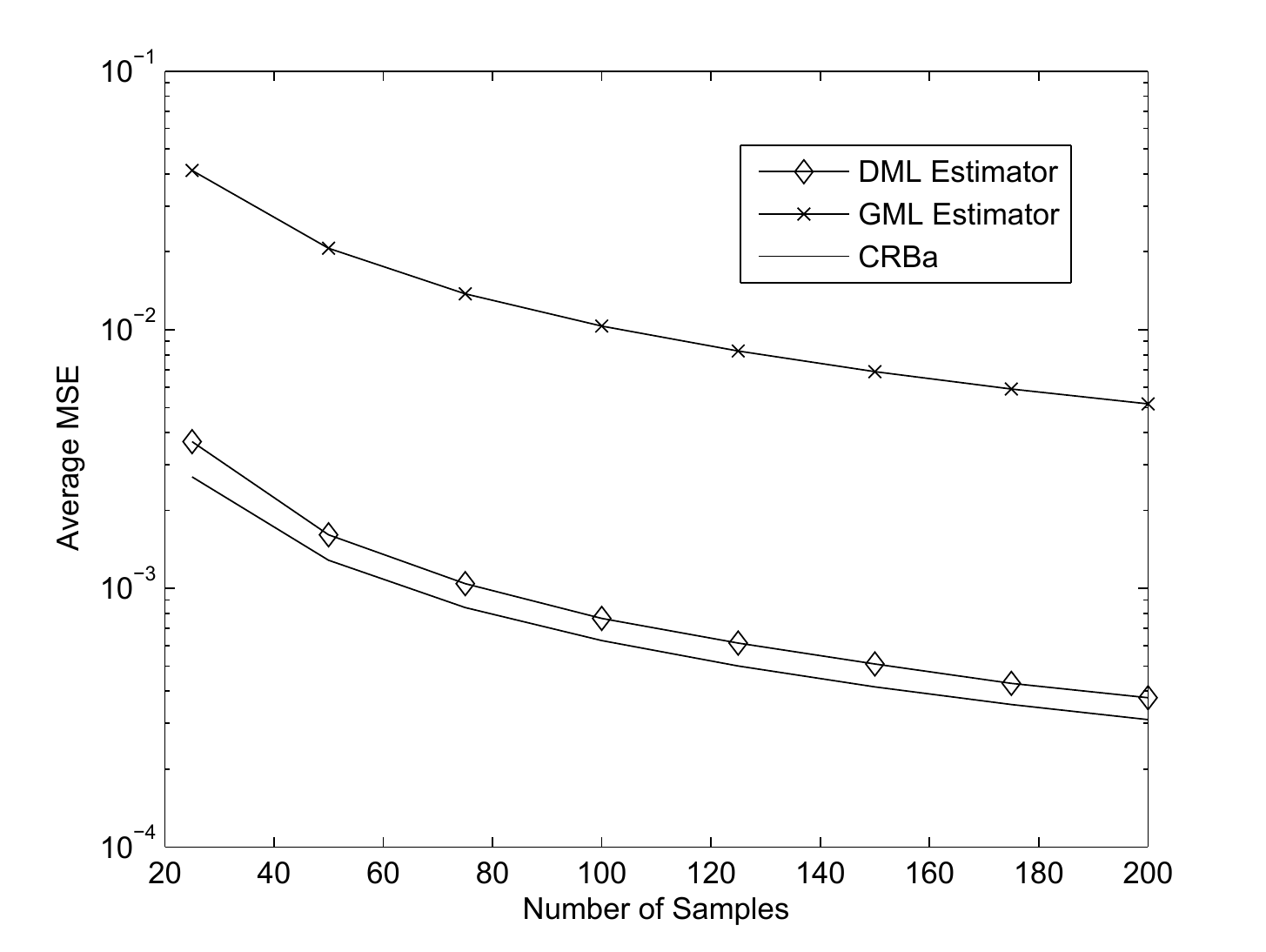} 
\caption{Average MSE of the DML and GML algorithms for the estimation of $a$ and the bound $CRBa$ plotted versus $N$ for $M=4$ and an SNR of $20$ dB.}
\label{MSE_QPSK_N}
\end{figure}

 \begin{figure}[!ht]
\centering
\includegraphics[width=4.4in, height=3.3in]{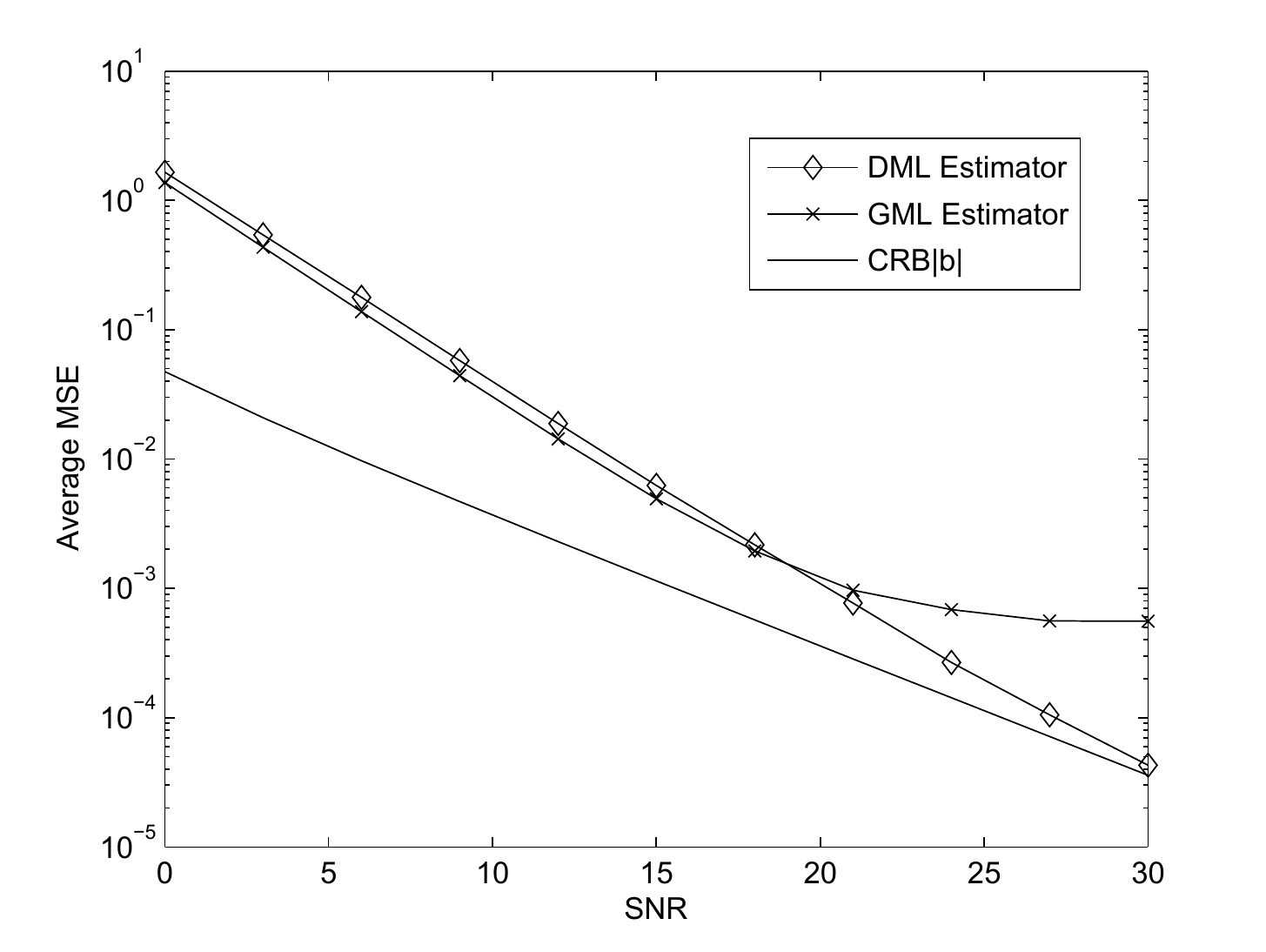}
\caption{Average MSE of the DML and GML algorithms for the estimation of $|b|$ and the bound $CRB_{|b|}$ plotted versus SNR for $M=4$ and $N=45$.}
\label{MSE_bQPSK_SNR}
\end{figure}

 \begin{figure}[!ht]
\centering
\includegraphics[width=4.4in, height=3.3in]{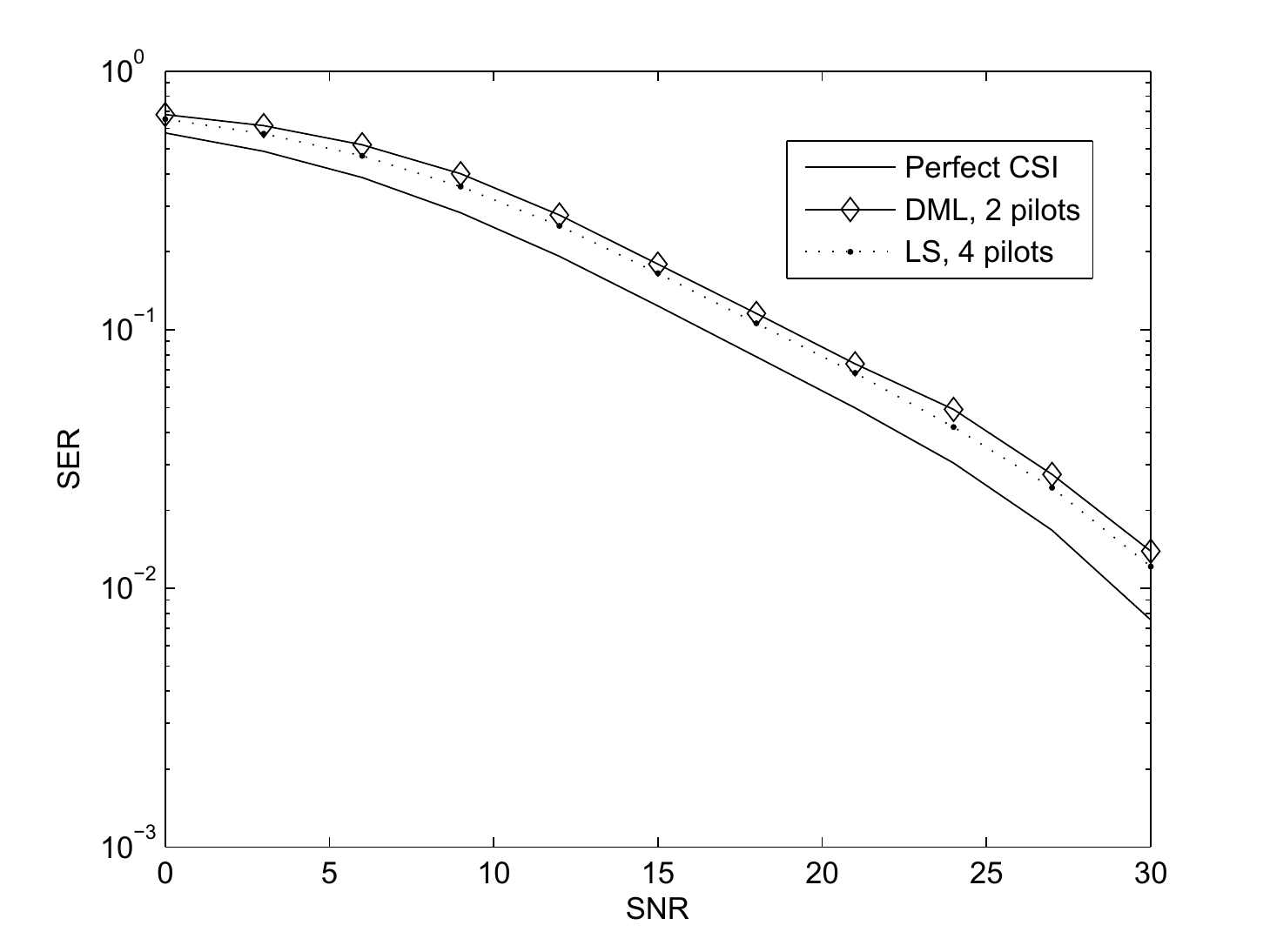} 
\caption{Average SER versus SNR for the DML and LS estimators for $M=4$, assuming the channel is fixed for the duration of 20 samples. We use 2 pilots to resolve the $M$-fold ambiguity in $\phi_b$ for the DML estimator, and 4 pilots to estimate $a$ and $b$ for the LS estimator. Also plotted is the SER performance under perfect CSI.}
\label{SER_QPSK1}
\end{figure}

 \begin{figure}[!ht]
\centering
\includegraphics[width=4.4in, height=3.3in]{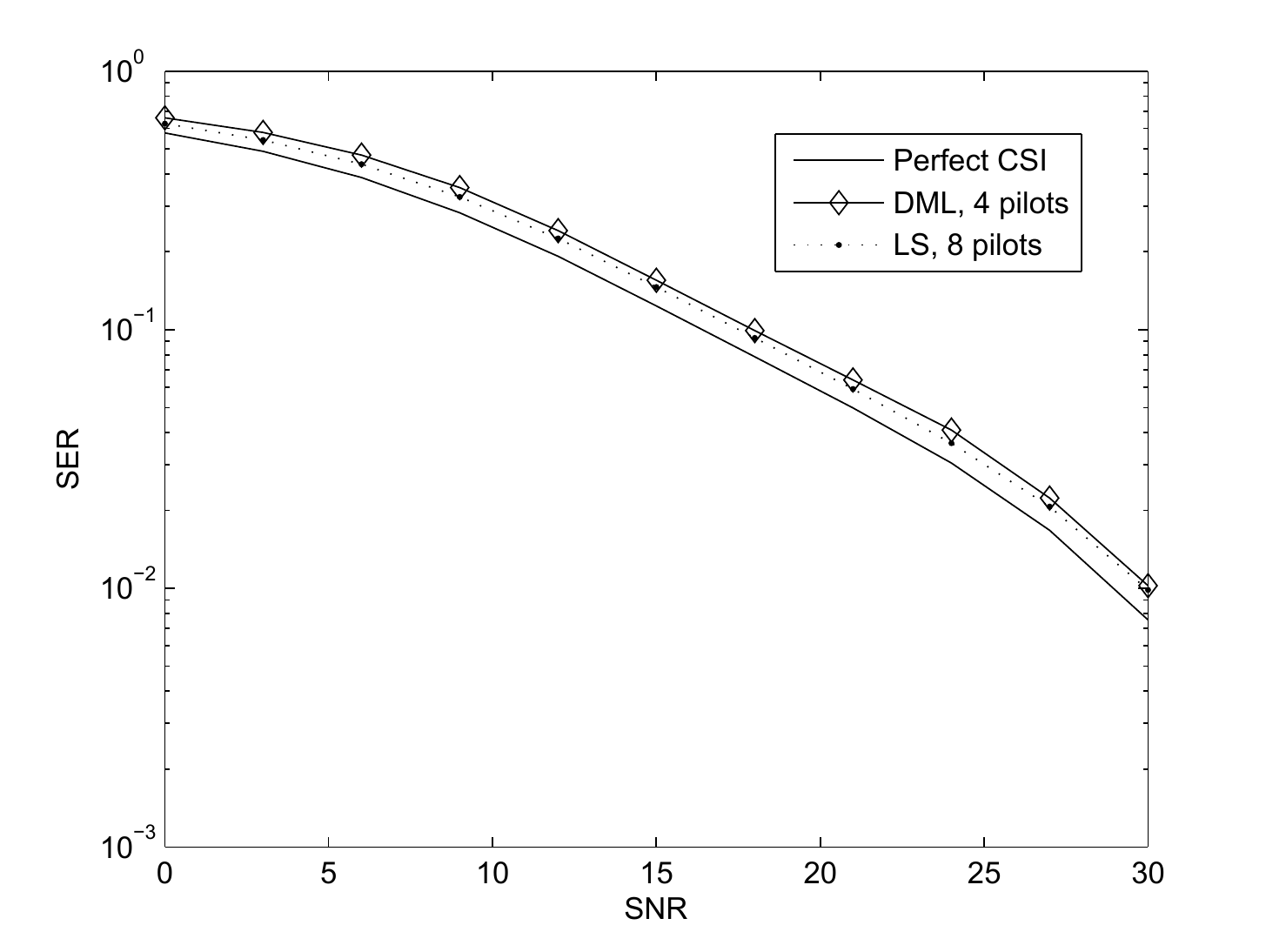}
\caption{Average SER versus SNR for the DML and LS estimators for $M=4$, assuming the channel is fixed for the duration of 40 samples. We use 4 pilots to resolve the $M$-fold ambiguity in $\phi_b$ for the DML estimator, and 8 pilots to estimate $a$ and $b$ for the LS estimator. }
\label{SER_QPSK2}
\end{figure}

\end{document}